\documentclass[journal,final]{IEEEtran}
\usepackage[T1]{fontenc}
\usepackage{bm}
\usepackage{amsmath}
\usepackage{amsthm}
\usepackage{amssymb}
\usepackage{graphicx}

\usepackage{stfloats}
\usepackage{subfigure} 
\usepackage{booktabs}
\usepackage{multicol}
\usepackage{multirow}
\usepackage{makecell}
\usepackage{color}

\usepackage[unicode=true,
bookmarks=true,bookmarksnumbered=true,bookmarksopen=true,bookmarksopenlevel=1,
breaklinks=false,pdfborder={0 0 0},pdfborderstyle={},backref=false,colorlinks=false]
{hyperref}
\hypersetup{pdftitle={Your Title},
	pdfauthor={Your Name},
	pdfpagelayout=OneColumn, pdfnewwindow=true, pdfstartview=XYZ, plainpages=false}

\makeatletter
\theoremstyle{plain}
\newtheorem{thm}{\protect\theoremname}
\theoremstyle{definition}

\theoremstyle{remark}
\newtheorem{rem}[thm]{\protect\remarkname}
\theoremstyle{plain}

\theoremstyle{plain}
\newtheorem{prop}[thm]{\protect\propositionname}

\usepackage[vlined,boxed,ruled,linesnumbered,resetcount]{algorithm2e}
\usepackage{algpseudocode}
\usepackage[noadjust,sort,compress]{cite}

\providecommand{\definitionname}{Definition}
\providecommand{\lemmaname}{Lemma}
\providecommand{\propositionname}{Proposition}
\providecommand{\remarkname}{Remark}
\providecommand{\theoremname}{Theorem}
\usepackage[justification=centering]{caption}

\newcommand {\aplt} {\ {\raise-.5ex\hbox{$\buildrel<\over{\mbox{\scriptsize $\sim$}}$}}\ }

\usepackage{color, colortbl}
	
\definecolor{Gray}{gray}{0.9}
\usepackage[first=0,last=9]{lcg}

\makeatother
 
\begin{document}
\title{Lattice-Aided Extraction of Spread-Spectrum Hidden Data}

\author{Fan Yang, Shanxiang Lyu, Hao Cheng, Jinming Wen, Hao Chen
	\thanks{
	Fan Yang, Shanxiang Lyu, Hao Cheng, Jinming Wen and Hao Chen are with the College of Cyber Security/College of Information Science and Technology, Jinan
		University, Guangzhou 510632, China.
		Corresponding author: Shanxiang Lyu (Email: lsx07@jnu.edu.cn).
	}
   \thanks{This work was presented in part at the 2022 International Conference on Digital Forensics \& Cyber Crime (ICDF2C).}
}
\maketitle

\begin{abstract}  
This paper discusses the problem of extracting spread spectrum hidden data from the perspective of lattice decoding. Since the conventional blind extraction scheme multi-carrier iterative generalize least-squares (M-IGLS) and non-blind extraction scheme minimum mean square error (MMSE) suffer from performance degradation when the carriers lack sufficient orthogonality, we present two novel schemes from the viewpoint of lattice decoding, namely multi-carrier iterative successive interference cancellation (M-ISIC) and sphere decoding (SD). The better performance of M-ISIC and SD are confirmed by both theoretical justification and numerical simulations.
\end{abstract}

\section{Introduction}
\noindent

\IEEEPARstart{D}{ata}
hiding describes the process of embedding secret messages into different forms of multimedia and transmitting over the open channel. As an important complement to conventional cryptographic systems, it provides flexible solutions for copyright protection, integrity verification, covert communication and other information security fields. To meet the requirements of various scenarios, the researchers' goals include reducing the distortion of the cover to get imperceptibility, increasing hidden capacity, and improving the robustness of the embedding scheme.
 
Watermark embedding and extraction are two crucial parts in the data hiding model. There are many literature describe various data hiding schemes over the past three decades \cite{cox2002digital, 9053484, lin2021lattice, lu2020secure, du2020high}, one of the mainstream directions is spread-spectrum (SS) steganography, because it has good robustness and security. By introducing the principle analogous to spread-spectrum communication, the concept of SS in data hiding was fist proposed by Cox et al. \cite{cox1997secure}. The basis idea of SS in data hiding is to disperse the message into many frequency bins of the host data by pseudorandom sequences, so as to make the energy in each one extremely small and certainly undetectable. This is similar to transmitting a narrowband signal with a much larger bandwidth and a lower power density. Some schemes have been proposed to improve upon SS. E.g, using the technique of minimum-mean-square error to reduce the interference caused by the host itself \cite{malvar2003improved}, improving signature design to reduce the decoding error rate \cite{9053484}, and using multi-carriers instead of a single carrier to improve the number of payloads \cite{li2013extracting,li2015steganalysis}.
 
 Depending on whether the receiver has the 
 pre-shared keys, the extraction of information from the multicarrier SS watermarking system consists of blind and non-blind extractions. 
 
\textbf{ i)} Blind extraction amounts to steganalysis via    “Watermarked Only Attack (WOA)” \cite{perez2009spread}. It is one of the scenarios that has attracted a lot of attention since it models most of the practical problems. It assumes that the attacker only has access to the composed signal, without any information about the host data and the spreading codes. Under this premise, the process of fully recovering embedded data is called blind extraction. To break the single-carrier SS method, Gkizeli et al. \cite{1530251} proposed a blind method named iterative generalized least squares (IGLS) to recover unknown messages hidden in image, which has remarkable recovery performance and low complexity. However, steganographers may prefer multi-carrier SS transform-domain embedding to improve security or the amount of information in a single transmission. The steganalysis for this situation seems more worthy of study. Since the underlying mathematical problem of extracting multiple message sequences from a mixed observation is akin to blind source separation (BSS) in speech signal processing, classical BSS algorithms such as independent component analysis (ICA) \cite{bingham2000fast} and Joint Approximate Diagonalization of Eigenmatrix (JADE) \cite{cardoso1999high} can also be used to extract the hidden data. Regrettably, these algorithms are far from satisfactory due to the correlated signal interference caused by the multi-carrier SS problem. In this regard, 
Li Ming et al. \cite{li2013extracting} developed an improved IGLS scheme referred to as multi-carrier iterative generalized least-squares (M-IGLS). The crux inside M-IGLS is a linear estimator referred to as zero-forcing (ZF) in lattice decoding literature \cite{ling2011proximity,LyuL19TSP}. M-IGLS exhibits satisfactory performance only when the carriers/signatures show sufficient orthogonality. For instance, M-IGLS  shows the case of embedding (and extracting) $4$ data streams by modifying $63$ host coefficients \cite{li2013extracting}.



\textbf{ii)} Non-blind extraction of SS watermarking adopts linear minimum mean square error (MMSE) estimator as the default option \cite{ingle2000statisical,li2013extracting}. However, linear MMSE is optimal only when the prior symbols admit Gaussian distributions, rather than the discrete distribution over $\lbrace\pm 1\rbrace$ \cite{LyuL19TSP}. MMSE also works well when the embedding matrix defined by carriers features sufficient orthogonality, but this property may not be guaranteed in the transmitter's side. 


As the discrete symbols (i.e., $\lbrace\pm 1\rbrace$) in multicarrier SS naturally induces lattices, it becomes tempting to adopt more sophisticated lattice decoding algorithms to improve upon the blind and non-blind extraction of  multicarrier SS watermarking. For this reason, this paper contributes in the following aspects:
\begin{itemize}
	\item  First, we propose a new hidden data blind extraction algorithm referred to as multi-carrier iterative successive interference cancellation (M-ISIC). Like M-IGLS, M-ISIC also estimates the mixing matrix and the integer messages iteratively by alternating minimization principle. However, in the step when the mixing matrix has been estimated, M-ISIC adopts successive interference cancellation (SIC) rather than ZF. Due to the larger decoding radius of SIC over ZF, the proposed M-ISIC is deemed to enjoy certain performance gains. Moreover, M-ISIC also features low complexity.
 \item  Second, we present a sphere decoding (SD) algorithm for the legit extraction of multicarrier SS watermarking. While maximum-likelihood (ML) extraction can be implemented via a brute-force enumeration,  sphere decoding (SD) \cite{fincke1985improved}
 is the better implementation of ML to save computational complexity. The magic of SD is to restrict the search space to within a sphere enclosing the query vector. Simulation results show that SD outperforms the default MMSE estimator especially when the channel matrix lacks sufficient orthogonality. 
 \item  Third, by formulating the problem of extracting multi-carrier SS as a lattice decoding problem, it fosters a deeper connection between the data hiding community and the post-quantum cryptography community. Lattice-based constructions are currently important candidates for post-quantum cryptography \cite{chinaf/WangLLWZLCYWL22}. The analysis of the security level of lattice-based cryptographic schemes also relies on sophisticated lattice decoding algorithms. This implies that, in the future, a novel algorithm for one community may also be explored for the other.
\end{itemize}

The rest of this paper is organized as follow. In Section \uppercase\expandafter{\romannumeral2}, preliminaries on SS embedding and lattice decoding are briefly introduced. In Section \uppercase\expandafter{\romannumeral3}, M-ISIC is presented and the comparisons between M-IGLS and M-ISIC are made.  Section \uppercase\expandafter{\romannumeral4} discusses sphere decoding and MMSE. Simulation results and conclusion are given in Section \uppercase\expandafter{\romannumeral5} and Section \uppercase\expandafter{\romannumeral6} respectively.

The following notation is used throughout the paper. Boldface upper-case   and lower-case letters represent matrices and column vectors, respectively. 
$\mathbb{R}$ denotes the set of real numbers, while $\mathbf{I}$ denotes an identity matrix. $(\cdot)^\top$ is the 
matrix transpose operator, and  $||\cdot||$, $||\cdot||_F$ denote vector
norm, and matrix Frobenius norm, respectively.
$\mathrm{sign}(\cdot)$ represents a  quantization function  with respect to $\lbrace-1,1\rbrace$.

\section{Preliminaries}
\subsection{Basics of Multicarrier SS}
\subsubsection{Embedding}
Without loss of generality,  a standard gray-scale image $\mathbf{H}\in \mathcal{M}^{N_1 \times N_2}$  is chosen as the host, where $\mathcal{M}$ denotes a finite image alphabet and $N_1\times N_2$ denotes the size of the image. Then $\mathbf{H}$ is partitioned into $M$ non-overlapping blocks $\mathbf{H}_1, ..., \mathbf{H}_M$ (of size $\frac{{N}_{1}\times{N}_{2}}{{M}}$). After performing DCT transformation and zig-zag scanning for each block, the cover object in each block can be generated as $\mathbf{x}(m)\in \mathbb{R}^{L}$, where $L \leq\frac{{N}_{1}\times{N}_{2}}{M}$ and $m=1, ..., M$. 


Multicarrier SS embedding scheme employs $K$ distinct carriers (signatures) $\mathbf{s}_1, ..., \mathbf{s}_K$ to implant $K$ bits of information $b_1, ..., b_K \in \lbrace \pm 1 \rbrace$ to each $\mathbf{x}(m)$.  Subsequently, the modified cover (stego) is generated by
\begin{equation}\label{emd_equa}
\mathbf{y}(m) = \sum_{k=1}^{K} A_k b_k(m)\mathbf{s}_k +\mathbf{x}(m) + \mathbf{n}(m), \,\, m=1,2, ..., M,
\end{equation}
where $A_k$ denotes the embedding amplitude of $\mathbf{s}_k$,  $b_k(m)$  denotes the messages of the $m$th block, and $\mathbf{n}(m)$ represents the additive white Gaussian noise vector of mean $\mathbf{0}$ and covariance $\sigma_{n}^2 \mathbf{I}_L$. For symbolic simplicity, we can express the embedding of $\mathbf{b}(1), ..., \mathbf{b}(M)$ in the matrix form as
\begin{equation}
\mathbf{Y}=\mathbf{VB}+\mathbf{Z}\label{ssembed},
\end{equation}where $\mathbf{Y} \triangleq \left[\mathbf{y}(1), ...,  \mathbf{y}(M)\right]  \in\mathbb{R}^{L \times M}$, $\mathbf{B} \triangleq \left[ \mathbf{b}(1), ..., \mathbf{b}(M)\right]\in \lbrace \pm 1 \rbrace^{K\times M}$, $\mathbf{V} \triangleq \left[A_1\mathbf{s}_1, ...,  A_M\mathbf{s}_M\right]  \in\mathbb{R}^{L \times K}$,  $\mathbf{Z} \triangleq \left[\mathbf{x}(1) + \mathbf{n}(1), ...,  \mathbf{x}(M) + \mathbf{n}(M)\right]  \in\mathbb{R}^{L \times M}$. In general,   $K\leq L$, which avoids inducing underdetermined system of equations.

By taking expectation over the randomness of $\mathbf{s}_k$, the embedding distortion due to  $A_k b_k(m)\mathbf{s}_k$ is
\begin{equation}
D_k= \mathbb{E}\lbrace || A_k b_k(m)\mathbf{s}_k||^2 \rbrace = A_k^2, \,\, k=1,2, ..., K.
\end{equation}
Based on the statistical independence of signatures $\mathbf{s}_k$, the averaged total distortion  per block is defined as $D=\sum_{k=1}^{K} D_k=\sum_{k=1}^{K}A_k^2.$

\subsubsection{Legitimate Extraction}
In the receiver's side, with the knowledge of secrets/carriers $\mathbf{s}_k$ legitimate users can obtain high-quality embedded bit recovery of messages $b_k(m)$. By the auto-correlation matrix of host data and noise, we can define the auto-correlation matrix of observation data $\mathbf{Y}$ as following form
\begin{equation}
\mathbf{R}_{\mathbf{y}}= \mathbf{R}_{\mathbf{x}} + \sum_{k=1}^{K}A_k^2 \mathbf{s}_k\mathbf{s}_k^\top + \sigma_{n}^2 \mathbf{I}_L.\label{ry-1}
\end{equation}
For easy of analysis, equation \eqref{ry-1} can be further written as 
\begin{equation}
\mathbf{R}_{\mathbf{y}}=\frac{1}{M}\mathbf{X}\mathbf{X}^\top+\mathbf{V}\mathbf{V}^\top+\sigma_{n}^2 \mathbf{I}_L
\end{equation}
where $\mathbf{V}=(A_1\mathbf{s}_{1},A_2\mathbf{s}_{2},\cdots,A_K\mathbf{s}_{K})$ and $\mathbf{R}_{\mathbf{x}}=\frac{1}{M}\mathbf{X}\mathbf{X}^\top$.

The linear MMSE detector  has the capable of minimizing the mean square error between the true values and estimated values by taking into account the trade-off between noise amplification and interference suppression \cite{1285069}. Via the linear MMSE filter, the embedded symbols are estimated by 
\begin{equation}
\hat{\mathbf{B}}_{MMSE}= \mathrm{sign}\lbrace (\mathbf{V}^\top \mathbf{R}_{\mathbf{y}}^{-1}) \mathbf{Y}\rbrace. \label{mmse}
\end{equation}
Using sample averaging for $M$ received vectors, the estimate of $\mathbf{R}_\mathbf{y}$ $\hat{\mathbf{R}}_{\mathbf{y}}=\frac{1}{M}\sum_{m=1}^{M}\mathbf{y}\mathbf{y}^\top$ can be obtained. Replace $\mathbf{R}_{\mathbf{y}}$ in \eqref{mmse} with $\hat{\mathbf{R}}_{\mathbf{y}}$, we get sample-matrix-inversion MMSE (SMI-MMSE) detector \cite{ingle2000statisical}.

\subsection{Basics of Lattices}

\subsubsection{Lattice Decoding Problem}
 Lattices are discrete additive subgroups over $m$-dimensional Euclidean space $\mathcal{R}^{m}$, which can be defined as the integer coefficients liner combination of $n$ linearly independent vectors
\begin{equation}
\mathcal{L}( \mathbf{G} )=\left\{ \sum_{k=1}^{K} {x}_{k}\mathbf{g}_{k}\thinspace \mid \thinspace x_{k}\in\mathbb{Z}\right\}
\end{equation}
where $\mathbf{G} \triangleq [\mathbf{g}_{1},...,\mathbf{g}_{K}]$ is called a lattice basis. 

Computationally hard problems can be defined over lattices. The one related to this work is called the  closest vector problem (CVP) \cite{micciancio2002complexity}: given a query vector $\mathbf{t}$, it asks to find the closest vector to $\mathbf{t}$ from the set of lattice vectors $\mathcal{L}(\mathbf{G})$. Let the closest vector be $\mathbf{G}\mathbf{x}$, $\mathbf{x}\in \mathbb{Z}^K$, then we have
\begin{equation}\label{CVPeq}
\| \mathbf{G}\mathbf{x}-\mathbf{t} \| \leq \|\mathbf{G}\tilde{\mathbf{x}}-\mathbf{t}\|,\, \forall \tilde{\mathbf{x}} \in\mathbb{Z}^{K}.
\end{equation}
In general solving CVP for a random lattice basis incurs exponential computational complexity in the order of $\mathcal{O}(2^K)$, but for lattice basis whose $\mathbf{g}_{1},...,\mathbf{g}_{K}$ are close to being orthogonal, fast low-complexity algorithm can approximately achieve the performance of maximum likelihood decoding. 

\subsubsection{Lattice Decoding Algorithms}
Zero-forcing (ZF) \cite{ling2011proximity} and successive interference cancellation (SIC) \cite{lyu2021lattice} are fast low-complexity algorithms to detect the transmitted signals at the receiving end. The former obtains the output by multiplying the pseudo-inverse of $\mathbf{V}$ to the left of $\mathbf{Y}$. The latter introduces decision feedback to decode each symbol successively, achieving better performance than the former. 

\begin{figure}[h]
	\centering
	\includegraphics[width=0.38\textwidth]{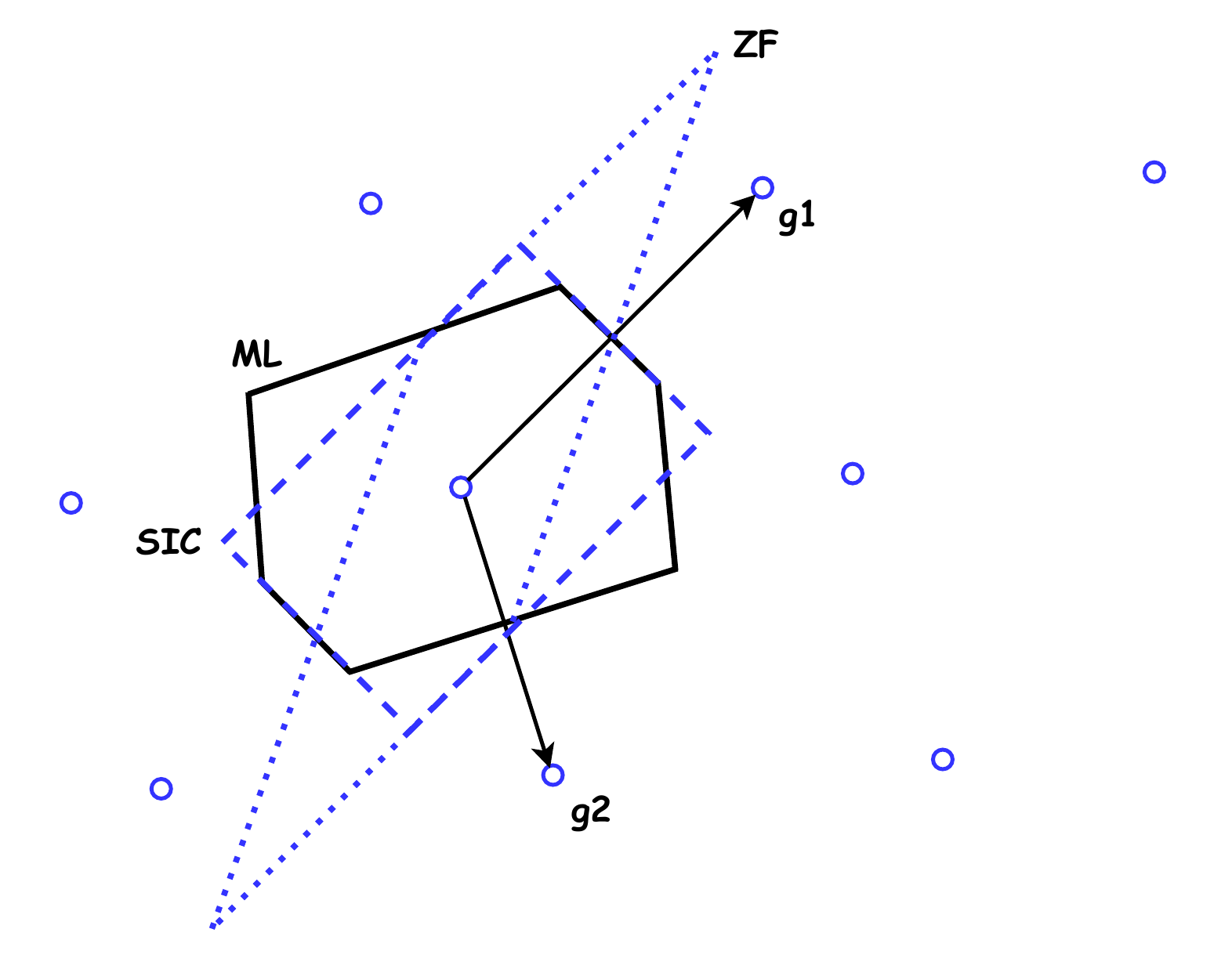}
	\caption{The decision regions of ZF (parallelogram) and SIC (rectangle) in a 2-dimensional lattice.}
	\label{desreg}
\end{figure} 

Fig. \ref{desreg} plots the decision boundaries for ZF and SIC. The elongated and narrow parallelogram is the decision region of ZF. Because the basis vectors are highly correlated, a slight perturbation of the noise can lead to a detection error. For SIC, the decision region is rectangle as only one symbol is decoded at a time \cite{1188114}.
Both ZF and SIC have worse performance than the optimal maximum-likelihood (ML) estimation due to their inherent nature of polynomial complexity.  More comparisons of ZF and SIC are presented in the section of blind extraction, while the application of SD will be addressed in the non-blind extraction.
  
\section{Blind Extraction}
The task of blind extraction requires estimating both $\mathbf{V}$ and $\mathbf{B}$ from the observation $\mathbf{Y}$, which is known as the noisy BSS problem:
\begin{equation}\label{eq_ob2}
\mathcal{P}_{1}: \min_{\mathbf{B} \in \{\pm 1\}^{K \times M} \atop \mathbf{V}\in\mathbb{R}^{L \times K}} || \mathbf{R}_{\mathbf{z}}^{-\frac{1}{2}}\mathbf{(Y-VB)} ||^{2}_{F},
\end{equation}
where $\mathbf{R}_{\mathbf{z}} \triangleq \mathbf{R}_{\mathbf{x}} + \sigma_{n}^2 \mathbf{I}_L$ denotes the pre-whitening matrix. Nevertheless, enumerating all the feasible candidates of $\mathbf{V}$ and $\mathbf{B}$ is infeasible as it incurs  exponential complexity. 


In the following, we briefly  describe the   M-IGLS that was proposed in \cite{li2013extracting} to solve $\mathcal{P}_{1}$. Then we improve the ZF detector in M-IGLS from the viewpoint of lattices.

\subsection{M-IGLS}
The pseudo-code of  M-IGLS is shown in Algorithm \ref{IGLS}.
Specifically, M-IGLS estimates  $\mathbf{V}$ and $\mathbf{B}$  iteratively by using an MMSE criterion: by either fixing $\mathbf{B}^{(d)}$  or $\mathbf{V}^{(d)}$ and using convex optimization, the formulas for $\mathbf{B}^{(d)}$  or $\mathbf{V}^{(d)}$ are derived.

\begin{algorithm}[t!] 
	\KwIn{$\mathbf{Y}$, $\mathbf{R}_{\mathbf{y}}$.} 
	\KwOut{$\hat{\mathbf{V}}=\mathbf{V}^{(d)}$, $\hat{\mathbf{B}}=\mathbf{B}^{(d)}$.}  
	$d=0$, $\mathbf{B}^{(0)}{\sim \{\pm 1\}^{K \times M}}$\;
	\While{ a stopping criterion has not been reached}{ $d\leftarrow d+1$ \;
		$\mathbf{V}^{(d)}\leftarrow\mathbf{Y}(\mathbf{B}^{(d-1)})^\mathrm{T}[\mathbf{B}^{(d-1)}(\mathbf{B}^{(d-1)})^\mathrm{T}]^\mathrm{-1}$\;
		$\mathbf{B}^{(d)}\leftarrow\mathrm{sign}\left\{\left((\mathbf{V}^{(d)})^\mathrm{T}\mathbf{R}_{\mathbf{y}}^{-1}\mathbf{V}^{(d)}\right)^\mathrm{-1}
		(\mathbf{V}^{(d)})^\mathrm{T}\mathbf{R}_{\mathbf{y}}^{-1}\mathbf{Y}\right\}$; \Comment{Approximate lattice decoding via GLS/ZF.}
	} 
	\caption{The M-IGLS data extraction algorithm.}
	\label{IGLS} 
\end{algorithm}
Observe the step of estimating $\mathbf{B}^{(d)}$ in Algorithm \ref{IGLS}, which asks to solve the following problem:
\begin{equation}\label{eq_ob3}
\mathcal{P}_{2}: \min_{\mathbf{B} \in \{\pm 1\}^{K \times M}} || \mathbf{R}_{\mathbf{z}}^{-\frac{1}{2}}\mathbf{Y} - \mathbf{R}_{\mathbf{z}}^{-\frac{1}{2}}\mathbf{V}\mathbf{B}||^{2}_{F}.
\end{equation}
Since $\{\pm 1\}^{K \times M} \in \mathbb{Z}^{K \times M}$,
$\mathcal{P}_{2}$ is a special case of CVP, which asks to find $M$ closest lattice vectors to $ \mathbf{R}_{\mathbf{z}}^{-\frac{1}{2}}\mathbf{Y}$, and the lattice is defined by basis $\mathbf{R}_{\mathbf{z}}^{-\frac{1}{2}}\mathbf{V}$. Considering $\mathcal{P}_{2}$, define the set of query vectors as $\overline{\mathbf{Y}} \triangleq \mathbf{R}_{\mathbf{z}}^{-\frac{1}{2}}\mathbf{Y}$, and the lattice basis as $\overline{\mathbf{V}} \triangleq \mathbf{R}_{\mathbf{z}}^{-\frac{1}{2}}\mathbf{V}$, then the ZF estimator is 
\begin{align}
\hat{\mathbf{B}}_{\mathrm{ZF}} & = \overline{\mathbf{V}}^{\dagger}\overline{\mathbf{Y}}^\mathrm{T} \nonumber \\
&	=  (\overline{\mathbf{V}}^\mathrm{T}\overline{\mathbf{V}})^{-1}\overline{\mathbf{V}}^\mathrm{T}\overline{\mathbf{Y}}^\mathrm{T}.
\end{align}

In Appendix A, we show that the geometric least square (GLS) step in line $5$ of M-IGLS is the same as ZF. The ZF estimator is linear, which behaves like a linear filter and separates the data streams and thereafter independently decodes each stream. The drawback of ZF is the effect of noise amplification when the lattice basis $\overline{\mathbf{V}}$ is not orthogonal. 

\subsection{M-ISIC}
By using decision feedback in the decoding process, the nonlinear Successive Interference Cancellation (SIC) detector has better performance than ZF. Recall that for $\mathcal{P}_{2}$, the lattice basis is $\overline{\mathbf{V}}$, and the set of query vectors are $\overline{\mathbf{Y}}$. The SIC algorithm consists of the following steps:

\noindent \textbf{Step i)}  
Use QR decomposition to factorize $\overline{\mathbf{V}}$: $\overline{\mathbf{V}}=\mathbf{QR}$ \footnote{For better performance, this paper adopts a sorted version of QR decomposition, where the column vectors in $\overline{\mathbf{V}}$ are sorted from short to long.}, where   $\mathbf{Q}\in\mathbb{R}^{L \times L}$ denotes a unitary matrix  and $\mathbf{R}\in\mathbb{R}^{L \times K}$ is an upper triangular matrix of the form:
\begin{equation}\label{Rdef}
\mathbf{R}= \begin{bmatrix} R_{1,1} & R_{1,2} & \cdots & R_{1,K}\\
0 & R_{2,2} & \cdots & R_{2,K}\\
\vdots & \vdots & \ddots & \vdots \\
0 & 0 & \cdots & R_{K,K}\\
0 & 0 & \cdots & 0\\
\vdots & \vdots & \ddots & \vdots \\
0 & 0 & \cdots & 0\\
\end{bmatrix}.
\end{equation}

\noindent \textbf{Step ii)}   Construct  $\mathbf{Y}^{'}=\mathbf{Q}^{\top}\overline{\mathbf{Y}}\in\mathbb{R}^{L \times M}$, which consists of vectors $\mathbf{y}^{'}(1), ...,  \mathbf{y}^{'}(M)$.

\noindent \textbf{Step iii)} For $m=1, ..., M$, generate the estimation as
\begin{align}
\hat{b}_K(m) &=\mathrm{sign}\left(\frac{y_K'(m)}{R_{K,K}}\right), \label{sic_s1}\\
\hat{b}_k(m) &= \mathrm{sign}\left(\frac{y_k'(m)- \sum_{l=k+1}^{K}R_{k,l}\hat{b}_l(m)}{R_{k,k}}\right), \label{sic_s2}
\end{align}
where $k= K-1, K-2, ..., 1$, and $y_k'(m)$ denotes the $k$th component of  $\mathbf{y}^{'}(m)$.

By substituting the Step 5 in Algorithm \ref{IGLS} with the SIC steps, we obtain a new algorithm referred to as multi-carrier iterative successive interference cancellation (M-ISIC). Its pseudo-codes are presented in Algorithm \ref{ISIC}.  Notably, $\mathbf{V}^{(d)}$ is estimated in the same way as that of M-IGLS, and the performance improvements rely on SIC decoding. The stopping criterion can be set as when $||\mathbf{B}^{(d)}-\mathbf{B}^{(d-1)}||^{2}_{F} < 10^{-5}$.
\begin{algorithm}[t!]
	\KwIn{$\mathbf{Y}$, $\mathbf{R}_{\mathbf{z}}$.} 
	\KwOut{$\hat{\mathbf{V}}=\mathbf{V}^{(d)}$, $\hat{\mathbf{B}}=\mathbf{B}^{(d)}$.}  
	$d=0$, $\mathbf{B}^{(0)}{\sim \{\pm 1\}^{K \times M}}$\;
	\While{ a stopping criterion has not been reached}{ $d\leftarrow d+1$ \;
		$\mathbf{V}^{(d)}\leftarrow\mathbf{Y}(\mathbf{B}^{(d-1)})^\mathrm{T}[\mathbf{B}^{(d-1)}(\mathbf{B}^{(d-1)})^\mathrm{T}]^\mathrm{-1}$\;
		Employ Steps i)-iii) of SIC to estimate $\mathbf{B}^{(d)}$ \Comment{Approximate lattice decoding via SIC.}
	} 
	\caption{The M-ISIC data extraction algorithm.}
	\label{ISIC} 
\end{algorithm}

\begin{rem}
	The rationale of SIC is explained as follows.
	When detecting multiple symbols, if one of them can be estimated first, the interference caused by the already decoded can be eliminated when solving another, so as to reduce the effective noise of the symbol to be solved and to improve the bit error rate performance. To be concise, 
	denote the observation equation corresponding to $\mathcal{P}_{2}$ as 
	\begin{equation}\label{SIC_equival}
	\overline{\mathbf{Y}} = \overline{\mathbf{V}}\mathbf{B} + \overline{\mathbf{Z}}, 
	\end{equation}
	with $\overline{\mathbf{Z}}$ being the effective noise. Then the multiplication of $\mathbf{Q}^{\top}$ to (\ref{SIC_equival})
	is simply a rotation, which maintain the Frobenius norm of the objective function:
	\begin{align}
	||\overline{\mathbf{Y}} - \overline{\mathbf{V}}\mathbf{B} ||^{2}_{F} &= ||\overline{\mathbf{Z}} ||^{2}_{F}\\
	&=  ||\mathbf{Q}^{\top} \overline{\mathbf{Z}} ||^{2}_{F}\\
	&=  	||\mathbf{Q}^{\top} \overline{\mathbf{Y}} -  {\mathbf{R}}\mathbf{B}||^{2}_{F}.
	\end{align}
	Regarding Step iii), $\hat{b}_K(m), ..., \hat{b}_1(m)$ are estimated in descending order because the interference caused by these symbols can be canceled. Moreover, the divisions of 
	$R_{K,K}, ..., R_{1,1}$  in Eqs. (\ref{sic_s1})  (\ref{sic_s2}) imply that the effective noise level hinges on the quality of $R_{K,K}, ..., R_{1,1}$.
\end{rem}

\subsection{Performance Analysis}
We show that M-ISIC theoretically outperforms M-IGLS,  as SIC has better decoding performance than ZF when approximately solving $\mathcal{P}_{2}$.  With a slight abuse of notations, $\mathcal{P}_{2}$ can be simplified as $M$ instances of the following observation:
\begin{equation}\label{eq4sicproof}
\mathbf{y} = \mathbf{R}'\mathbf{b}^*+\mathbf{z}
\end{equation}
where $\mathbf{y} \in \mathbb{R}^K$, $\mathbf{b}^* \in \{\pm 1\}^{K}$ is the transmitted message, $\mathbf{R}' \in \mathbb{R}^{K\times K}$ includes only the first $K$ rows of 
(\ref{Rdef}), and  we assume that $\mathbf{z}$ also admits a Gaussian distribution with  mean $\mathbf{0}$ and covariance $\sigma_{n}^2 \mathbf{I}_K$. Then the lattice decoding task becomes
\begin{equation}\label{eq_ob4}
\mathcal{P}_{\mathrm{3}}: \min_{\mathbf{b} \in \{\pm 1\}^{K}} || \mathbf{y} - \mathbf{R}\mathbf{b} ||^{2}.
\end{equation}

It has been demonstrated in the literature \cite{ling2011proximity,9497074} that SIC outperforms ZF if the constraint of $\mathbf{b}$ in $\mathcal{P}_{\mathrm{3}}$ is an integer set $\mathbb{Z}^K$ and a box-constrained (truncated continuous integer) set $\mathcal{B}$. Therefore, we employ a model reduction technique to show that SIC has higher success probability when decoding $\mathcal{P}_{\mathrm{3}}$.

\begin{prop}\label{prop2}
	Let the SIC and ZF estimates of $\mathcal{P}_{\mathrm{3}}$ be
	$\mathbf{b}^{\mathrm{SIC}}$ and $\mathbf{b}^{\mathrm{ZF}}$, respectively. Then the averaged decoding success probability of SIC is higher than that of ZF:
	\begin{equation}\label{PropSICzf}
	\mathbb{E}_{\mathbf{b}^*}\lbrace \mathrm{Pr}(\mathbf{b}^{\mathrm{SIC}}=\mathbf{b}^*)  \rbrace \geq \mathbb{E}_{\mathbf{b}^*}\lbrace \mathrm{Pr}(\mathbf{b}^{\mathrm{ZF}}=\mathbf{b}^*) \rbrace,
	\end{equation}
	where the expectation is taken over uniform random $\mathbf{b}^* \in \{\pm 1\}^{K}$.
\end{prop}
\begin{proof}
	Firstly, Eq. (\ref{eq4sicproof}) is rewritten as 
	\begin{equation}
	(\mathbf{y} + \mathbf{R}\times \mathbf{1})/2 = \mathbf{R}(\mathbf{b}^*+\mathbf{1})/2+\mathbf{z}/2.
	\end{equation}
	By updating the query vector $\mathbf{y}$ as $\mathbf{y}'\triangleq (\mathbf{y} + \mathbf{R}\times \mathbf{1})/2$, the bipolar constraint model $\mathcal{P}_{\mathrm{3}}$ is transformed to the following box-constrained model $\mathcal{P}_{\mathrm{4}}$:
	\begin{equation}\label{eq_ob5}
	\mathcal{P}_{\mathrm{4}}: \min_{\mathbf{b} \in \mathcal{B}} || \mathbf{y}' - \mathbf{R}\mathbf{b} ||^{2}, 
	\end{equation}
	where the constraint of the variable is $\mathcal{B}=\{0,1\}^{K}$. Since \cite{9497074}[Thm. 9] has shown that  Eq. (\ref{PropSICzf}) holds in this type of box-constrained model, the proposition is proved.
\end{proof}

If $\overline{\mathbf{V}}$ is close to being an orthogonal matrix, then ZF and SIC detection can both achieve maximum likelihood estimation. The reason is that they are all solving a much simpler quantization problem $ \min_{\mathbf{b} \in \{\pm 1\}^{K}} || \mathbf{y} - \mathbf{I}_K\mathbf{b} ||^{2}.$ In general, the performance gap between ZF and SIC depends on  the degree of orthogonality of the lattice basis  $\overline{\mathbf{V}}$. To quantify this parameter, we introduce the normalized orthogonality defect of a matrix as 
\begin{equation}
\delta(\overline{\mathbf{V}})= \left(\frac{\prod_{k=1}^{K}||\overline{\mathbf{v}}_k||}{\sqrt{\det({\overline{\mathbf{V}}}^\top \overline{\mathbf{V}})}} \right)^{1/K},
\end{equation}
where the column vectors of $\overline{\mathbf{V}}=[\overline{\mathbf{v}}_1, ..., \overline{\mathbf{v}}_K]$ are linear independent.
From Hardamard’s inequality, $\delta(\overline{\mathbf{V}})$
is always larger than or
equal to $1$, with equality if and only if the columns are orthogonal to each other. Summarizing the above, SIC performs better than ZF in general, and their performance gap decreases as  $\delta(\overline{\mathbf{V}}) \rightarrow 1$.

 \subsection{Computational Complexity}
 To compare with M-IGLS and exiting schemes, we give the computational complexity of M-ISIC based on the following conditions:
 \begin{itemize}
 	\item The complexity of the multiplication of two matrices
 	$\mathbf{A}\in\mathbb{R}^{M\times N}$ and $\mathbf{B}\in\mathbb{R}^{N\times K}$ is $\mathcal{O}(MNK)$.
 	\item  The complexity of an inversion over the square matrix   $\mathbf{A}\in\mathbb{R}^{N\times N}$  is $\mathcal{O}(N^{3})$.
 	\item The complexity of performing QR decomposition on matrix $\mathbf{A}\in\mathbb{R}^{M\times N}$, $M > N$, is $\mathcal{O}(2MN^{2})$.
 \end{itemize}
 
 Notice that $\mathbf{Y}\in\mathbb{R}^{L\times M}$, $\mathbf{V}\in\mathbb{R}^{L\times K}$ and $\mathbf{B}\in\mathbb{R}^{K\times M}$,  the computational complexity of Step 4 in M-ISIC is 
 \[\mathcal{O}(K^{3}+K^{2}(L+M)+LMK).\]
 The   computational complexity of Step 5 is dominated by the QR decomposition, which is 
 \[\mathcal{O}\left(K^{2}L+M(LK+K)\right).\]                
 The computational complexity of each iteration of the algorithm is summarized as
 \[ \mathcal{O}\left( K^{3}+2LMK+K^{2}(3L+M) +KM  \right).\]               
 With a total of $T$ iterations, the overall complexity is
 \[ \mathcal{O}\left( T(K^{3}+2LMK+K^{2}(3L+M) +KM ) \right).\]

\section{Non-blind Extraction}
The difference between legitimate/non-blind extraction and blind extraction lies in the availability of $\mathbf{V}$. In the case of 
legitimate extraction, it asks to solve
\begin{equation}\label{eq_ob6}
\mathcal{P}_{4}: \min_{\mathbf{B} \in \{\pm 1\}^{K \times M}} || \mathbf{R}_{\mathbf{z}}^{-\frac{1}{2}}(\mathbf{Y} - \mathbf{V}\mathbf{B})||^{2}_{F}.
\end{equation}  With the knowledge of carriers, non-blind algorithms exhibit higher accuracy than the blind algorithms. 

In this section, we describe the similarity between ZF and MMSE criterion by using an extended system model. To achieve better extraction performance when the channel matrix lacks sufficient orthogonality, we introduce a sphere decoding algorithm to extract SS hidden data. Subsequently, its computational complexity is discussed.

\subsection{Equivalence of linear MMSE and ZF}
 By introducing an extended system model, it is straightforward to show the similarity between linear MMSE and zero forcing (ZF). The channel matrix $\mathbf{V}$ and the received matrix $\mathbf{Y}$ can be reconstructed through
\begin{equation}
\underline{\mathbf{V}} = \begin{bmatrix} \mathbf{V}^\top \hspace{1mm} \sigma_{n} \mathbf{I}_L \hspace{1mm} \frac{1}{\sqrt{M}}\mathbf{X}^\top\\
\end{bmatrix}^\top \quad \text{and}  
\quad \underline{\mathbf{Y}} = \begin{bmatrix} \mathbf{Y}^\top \hspace{1mm} \textbf{0} \end{bmatrix}^\top.
\end{equation}
Therefore, the output of the linear MMSE filter \eqref{mmse} can be re-expressed as
\begin{align}
\hat{\mathbf{B}}_{MMSE} & = \mathrm{sign}\lbrace \underline{\mathbf{V}}^\top(\underline{\mathbf{V}}^\top\underline{\mathbf{V}})^{-1}\underline{\mathbf{Y}}\rbrace \\
& = \mathrm{sign}\lbrace \underline{\mathbf{V}}^\dagger\underline{\mathbf{Y}}\rbrace.\label{extend_zf}
\end{align} 
It is not difficult to find that \eqref{extend_zf} is analogous to the familiar linear zero forcing detector $\hat{\mathbf{B}}_{ZF} = \mathrm{sign}\lbrace\mathbf{V}^\dagger\mathbf{Y}\rbrace $ \cite{1285069}, except that $\mathbf{V}$ and $\mathbf{Y}$ are replaced by $\underline{\mathbf{V}}$ and $\underline{\mathbf{Y}}$ respectively. 

Since $\mathcal{P}_{4}$ amounts to the CVP of lattices, there is no free lunch for linear complexity algorithms (such as ZF,  SIC, linear MMSE) to provide high-quality or exact solutions for $\mathcal{P}_{4}$. The default linear MMSE algorithm shows satisfactory performance only when $\mathbf{V}$ has sufficient orthogonality (i.e., $\mathbf{V}^\top \mathbf{V}$ is close to an identity matrix). To solve $\mathcal{P}_{4}$ exactly, exponential-complexity algorithms are indispensable.

\subsection{Sphere Decoding}
\begin{figure}[t!]
	\centering
	\includegraphics[width=0.35\textwidth]{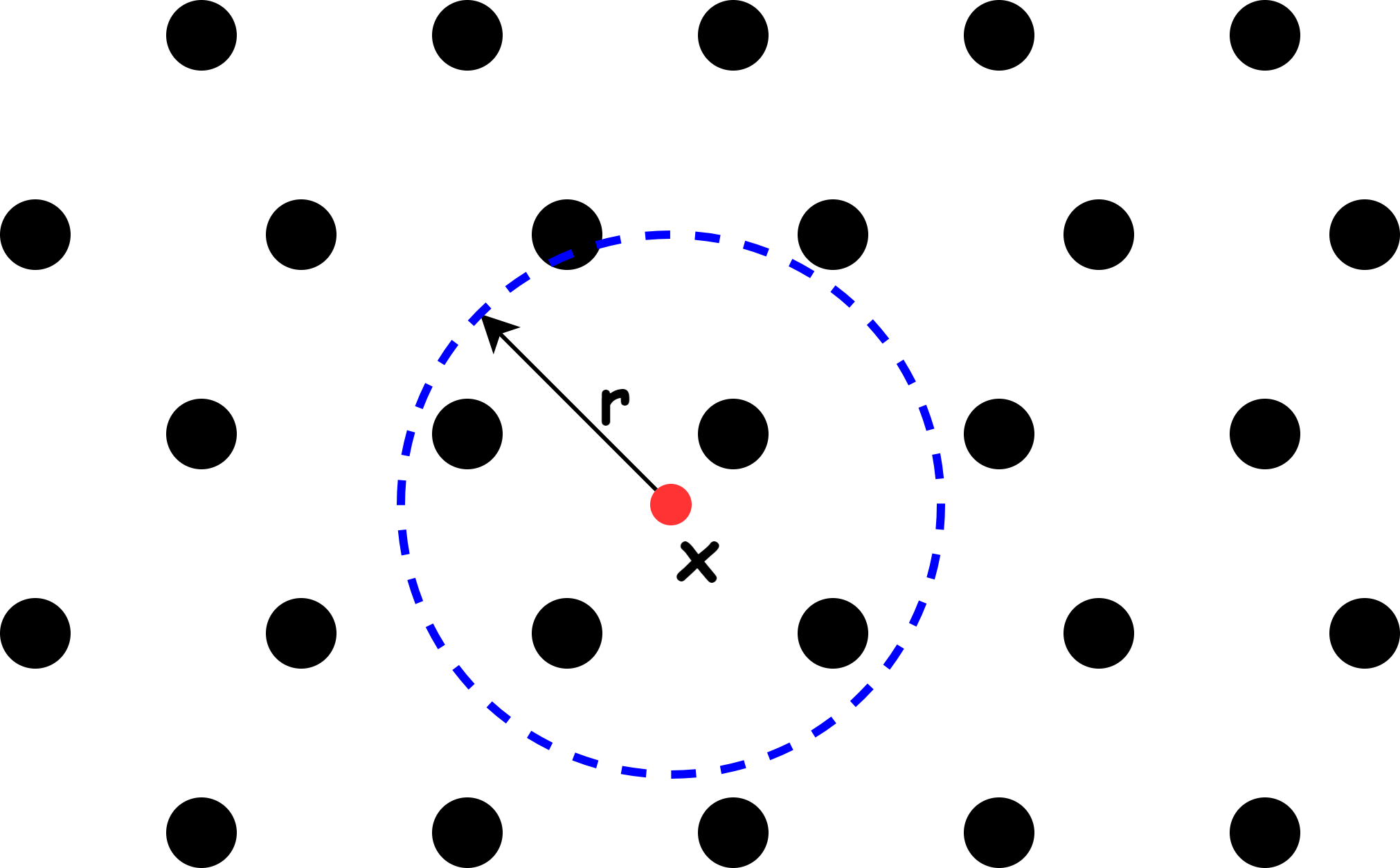}
	\caption{Sphere decoding in 2-dimensional Euclidean space.}
	\label{sd_r}
\end{figure} 
Sphere decoding is a popular approach to solve CVP in the lattice community \cite{micciancio2002complexity}. We only have to adjust a few steps of the conventional sphere decoding to solve $\mathcal{P}_{4}$: i) The quantization of symbols is to $\lbrace \pm 1 \rbrace$ rather than $\mathbb{Z}$. ii) We initialize the initial search radius of sphere decoding via the solution of linear MMSE.


The principle behind sphere decoding is quite simple: we attempt to search over those lattice points that lie inside a sphere with radius $r$ and find the nearest vector to the given query vector $\mathbf{x}$. An example in 2-dimensional space is given in Fig. \ref{sd_r}.  If the channel matrix $\overline{\mathbf{V}}$ represents lattice basis, then the lattice point $\overline{\mathbf{V}}\mathbf{b}$ is located in the sphere of radius $d$ and center $\mathbf{x}$ if and only if
\begin{equation}\label{d}
d^2 \geq \|\overline{\mathbf{y}} - \overline{\mathbf{V}}\mathbf{b}\|^2, \quad \mathbf{b}\in \{\pm 1\}^K
\end{equation}
where $\|\cdot\|$ represents Euclidean norm, and $\overline{\mathbf{y}}$ and $\mathbf{b}$ are the columns of $\overline{\mathbf{Y}}$ and $\mathbf{B}$ respectively. Although it seems complicated to determinate which lattice points are contained within the sphere in $m$-dimensional space,  it becomes effortless to do so when $m=1$. The reason is that the sphere degenerates into a fixed length interval in one-dimension. Then the lattice points are the integer values falling in the interval centered on $\mathbf{x}$. With this basic observation, we can generalize from $k$ dimension to $k+1$ dimension. Assume that the lattice points contained in the sphere with radius $r$ are obtained. Then, for the sphere with the same radius in $k+1$ dimension, the desirable values of the $k+1$ coordinate of these lattice points form a finite set or a fixed length interval. It implies that one can obtain the lattice points contained in the sphere with radius $r$ in $m$ dimension by means of solving all the lattice points in the sphere of the same radius from $1,\cdots,m-1$ dimension successively. 

Through the above brief introduction, the sphere decoding algorithm consists of the following steps: \\
\noindent \textbf{Step i)} 
Perform QR decomposition to factorize $\overline{\mathbf{V}}$: $\overline{\mathbf{V}}=\mathbf{QR}$.$\mathbf{Q}=\begin{bmatrix}
\mathbf{Q}_1 \hspace{0.6mm} \mathbf{Q}_2
\end{bmatrix}\in\mathcal{R}^{L \times L}$ denotes a unitary matrix with pair-wise orthogonal columns  and $\mathbf{R}\in\mathcal{R}^{L \times K}$ denotes an upper triangular matrix.

\noindent \textbf{Step ii)}  On the basis of QR decomposition, \eqref{d} can be rewritten as
\begin{align}\label{rewrit_d}
d^2 & \geq \|\overline{\mathbf{y}} - \begin{bmatrix}\mathbf{Q}_{1}\hspace{1mm}\mathbf{Q}_{2}\end{bmatrix}\begin{bmatrix}
\mathbf{R}_{1}\\
\textbf{0}
\end{bmatrix}\mathbf{b}\|^2=\|\begin{bmatrix}
\mathbf{Q}_{1}^{H}\\
\mathbf{Q}_{2}^{H}
\end{bmatrix}\overline{\mathbf{y}}-\begin{bmatrix}
\mathbf{R}_{1} \\ 
\textbf{0}
\end{bmatrix}\mathbf{b}\|^2 \\
& = \|\mathbf{Q}_{1}^{H}\overline{\mathbf{y}}-\mathbf{R}_{1}\mathbf{b}\|^2 + \|\mathbf{Q}_{2}^{H}\overline{\mathbf{y}}\|^2\label{dqr}
\end{align}
where $(\cdot)^{H}$ denotes Hermitian transpose. To simplify the symbol, let us define $\overline{d}^{2}=d^2-\|\mathbf{Q}_{2}^{H}\overline{\mathbf{y}}\|^2$ and $\overline{\mathbf{y}}=\mathbf{Q}_{1}^{H}\overline{\mathbf{y}}$ to represent \eqref{dqr} as
\begin{align}\label{dbar}
\overline{d}^2 &\geq \|\overline{\mathbf{y}}-\mathbf{R}_{1}\mathbf{b} \|^2.
\end{align}
In accordance with the upper triangular property of $\mathbf{R}_{1}$, the right hand side of \eqref{dbar} can be expanded to a polynomial
\begin{align}
\overline{d}^{2} &\geq (\overline{y}_{K}-R_{K,K}
b_{K})^2+(\overline{y}_{K-1}-\sum_{j=K-1}^{K}R_{K-1,j}b_{j})^2 \nonumber \\ 
& + \cdots + (\overline{y}_{1}-\sum_{j=1}^{K}R_{1,j}b_{j})^2. \label{expand}
\end{align}

\noindent \textbf{Step iii)} 
Provided that $\mathbf{Y}$ and $\mathbf{V}$ are known, then $\overline{\mathbf{y}}$ is also known for the receiver. By observing the right hand side of \eqref{expand}, it is straightforward to deduce that the first term of \eqref{expand} only hinges on $\{b_{K}\}$, while the second hinges on $\{b_{K},b_{K-1}\}$, and so on. If the admissible values of $b_{K}$ have been estimated, the decoder will exploit this set to further estimate $b_{K-1}$. Let's start from one dimension. The first necessary condition for $\mathbf{Vb}$ to fall in the sphere is $\overline{d}^2 \geq (\overline{y}_{K}-R_{K,K}b_{K})^2$, i.e., $b_{K} $ must meet
\begin{equation}\label{c1}
\lceil\dfrac{-\overline{d}+\overline{y}_{K}}{R_{K,K}}\rceil \leq b_{K} \leq \lfloor\dfrac{\overline{d}+\overline{y}_{K}}{R_{K,K}}\rfloor
\end{equation}
where $\lceil\cdot\rceil$ and $\lfloor\cdot\rfloor$ denote rounding up and rounding down, respectively. There is no doubt that $\eqref{c1}$ is definitely not sufficient enough. We need stronger constraints in order to keep the search space shrinking. For any $b_{K}$ satisfying \eqref{c1}, let $\overline{d}_{K}^2=\overline{d}^2-(\overline{y}_{K}-R_{K,K}
b_{K})^2$ and $\overline{y}_{K-1}^{'}=\overline{y}_{K-1}-R_{K-1,K}b_{K}$, the integer interval that $b_{K-1}$ belongs to can be found.
\begin{equation}\label{c2}
\lceil\dfrac{-\overline{d}_{K}+\overline{y}_{K-1}^{'}}{R_{K-1,K-1}}\rceil \leq b_{K-1} \leq \lfloor\dfrac{\overline{d}_{k}+\overline{y}_{K-1}^{'}}{R_{K-1,K-1}}\rfloor
\end{equation}
Similarly, the intervals that the remaining symbols $b_{K-2},\cdots,b_{1}$ belong to can be calculated in the same way recursively. Finally, we obtain all the candidate lattice points, the potential closest vectors to $\overline{\mathbf{y}}$, after the program terminates.

Algorithm \ref{msd} gives the pseudo-code of sphere decoding. The radius is initialized to the solution of linear MMSE. As the message space is $\{\pm 1\}$, the rounding operation simply invokes $sign(\cdot)$.

\begin{algorithm}[t!]
	\KwIn{$\overline{\mathbf{y}}=\mathbf{Q}_{1}^{H}\overline{\mathbf{y}}$, $\mathbf{R}_1$, $Radius$.} 
	\KwOut{$\hat{\mathbf{b}}$.}  
	
	$Initialization: K=size(\mathbf{R}_1,2), dist=0, k=K,  \hat{\mathbf{b}}=zeros(K,1)$\;
	
	\eIf{$k==K$}{$\overline{\mathbf{y}}'=\overline{\mathbf{y}}$;}{$\overline{\mathbf{y}}'=\overline{\mathbf{y}}-\mathbf{R}_{1}(:,k+1:end)*\hat{\mathbf{b}}(k+1:end)$;}
	
	$c=sign\left( \dfrac{\overline{\mathbf{y}}(k)}{\mathbf{R}_{1}(k,k)} \right), \hat{\mathbf{b}}(k)=c$;
	
	$\overline{d}^2$=$\left( \overline{\mathbf{y}}'(k)-\mathbf{R}_{1}(k,k:end)\hat{\mathbf{b}}(k:end)  \right)^2+dist$;
	
	\eIf{$\overline{d}^2 \leq Radius$}{go to 12;}{go to 19;}
	
	\eIf{$k==1$}{save $\hat{\mathbf{b}}$\;
		$Radius=\overline{d}^2;$}{$k=k-1$\; $dist=\overline{d}^2$\; go to 2;}
	
	$ci=c*(-1)$, $\hat{\mathbf{b}}(k)=ci$, go to 7;
	
	\caption{The M-SD data extraction algorithm.}
	\label{msd} 
\end{algorithm}

\begin{figure}[t!]
	\centering
	\includegraphics[width=0.48\textwidth]{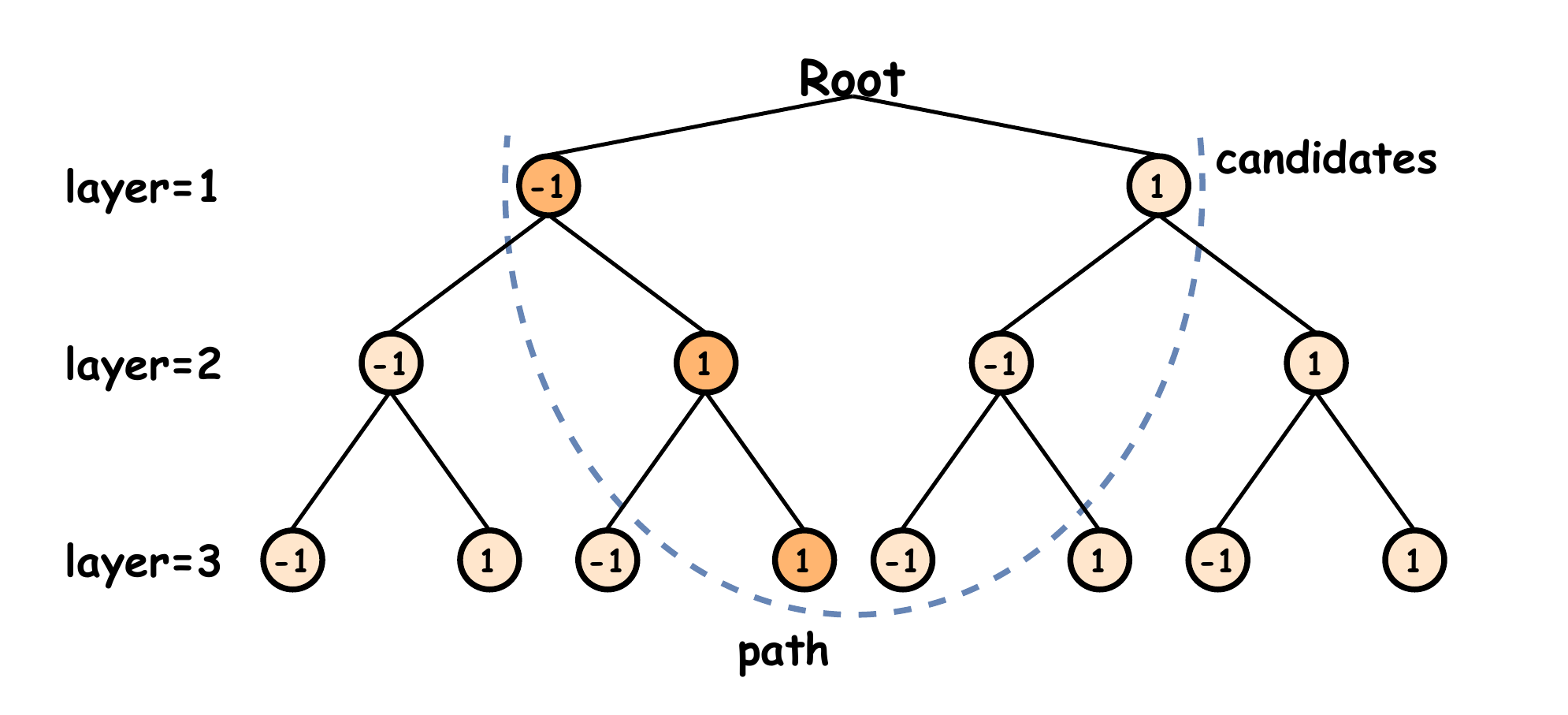}
	\caption{A 3-dimensional binary search tree.}
	\label{bina_tree}
\end{figure}

\subsection{Computational Complexity}
The computational complexity of sphere decoding has been studied thoroughly in the literature \cite{fincke1985improved,hassibi2005sphere}. In the worst case, we have to visit all $2^{K}$ nodes. But in generally the actual complexity is significantly smaller than that, as the algorithm is constantly updating the search radius. According to \cite{hassibi2005sphere}, the expected complexity of sphere decoding is of polynomial-time.   Fig. \ref{bina_tree} shows a binary search tree in 3-dimensional space, where the nodes in $k$-th layers correspond to the lattice points in $k$-dimension and the height of tree is $K$.
 
\begin{figure}[h!]
	\centering    
	\subfigure[No.0] { 
		\label{No.0}     
		\includegraphics[width=2.6cm]{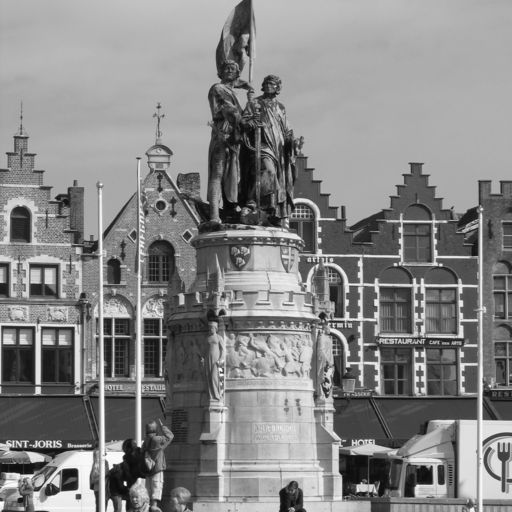}     
	}   
	\subfigure[No.709] {
		\label{No.709}     
		\includegraphics[width=2.6cm]{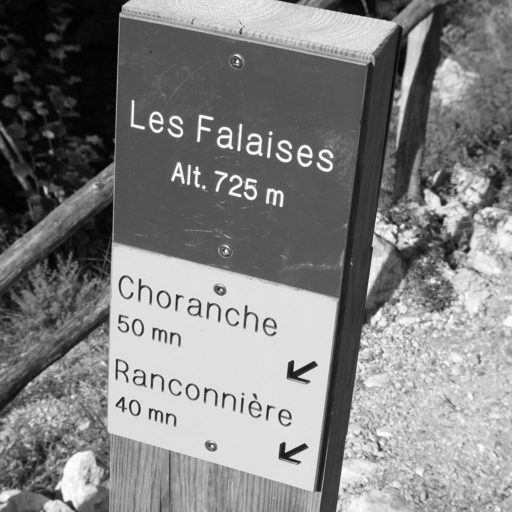}  
	}     
	\subfigure[No.5752] { 
		\label{No.5752}     
		\includegraphics[width=2.6cm]{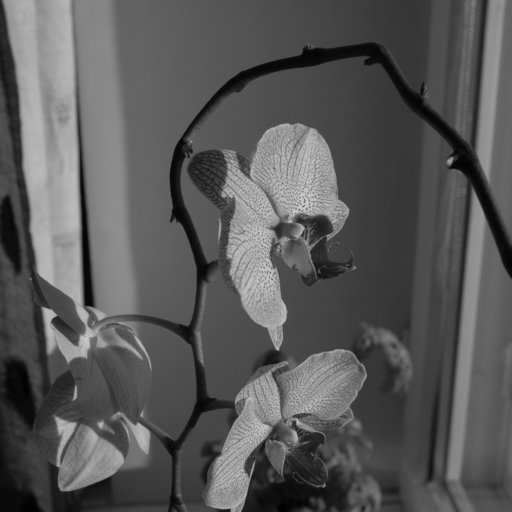}     
	}   
	\subfigure[No.6088] { 
		\label{No.6088}     
		\includegraphics[width=2.6cm]{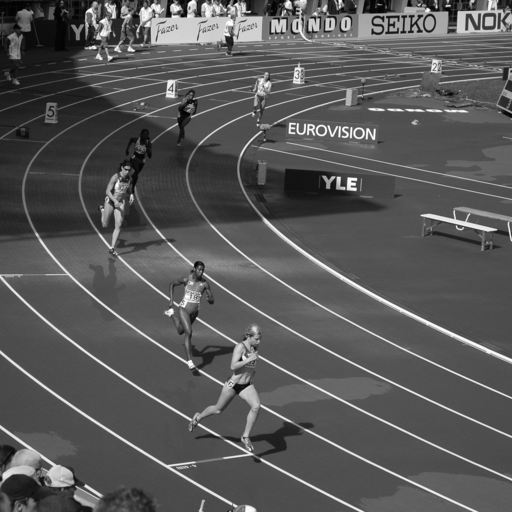}     
	}   
	\subfigure[No.6120] {
		\label{No.6120}     
		\includegraphics[width=2.6cm]{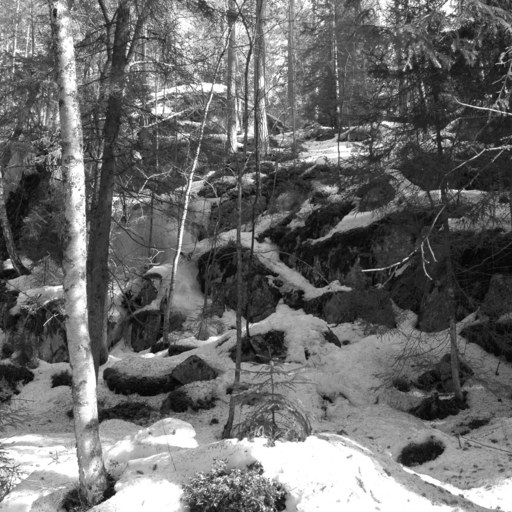}  
	}     
	\subfigure[No.6148] { 
		\label{No.6148}     
		\includegraphics[width=2.6cm]{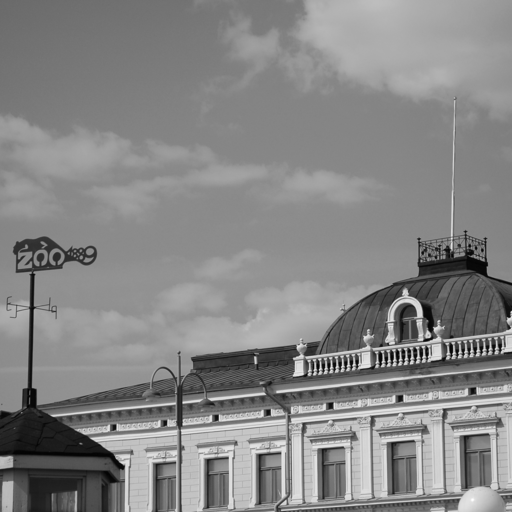}     
	} 
	\caption{Representative images in the BOWS-2 database.}     
	\label{stimage}     
\end{figure}

\begin{figure}[h!]
	\centering    
	\includegraphics[width=0.48\textwidth]{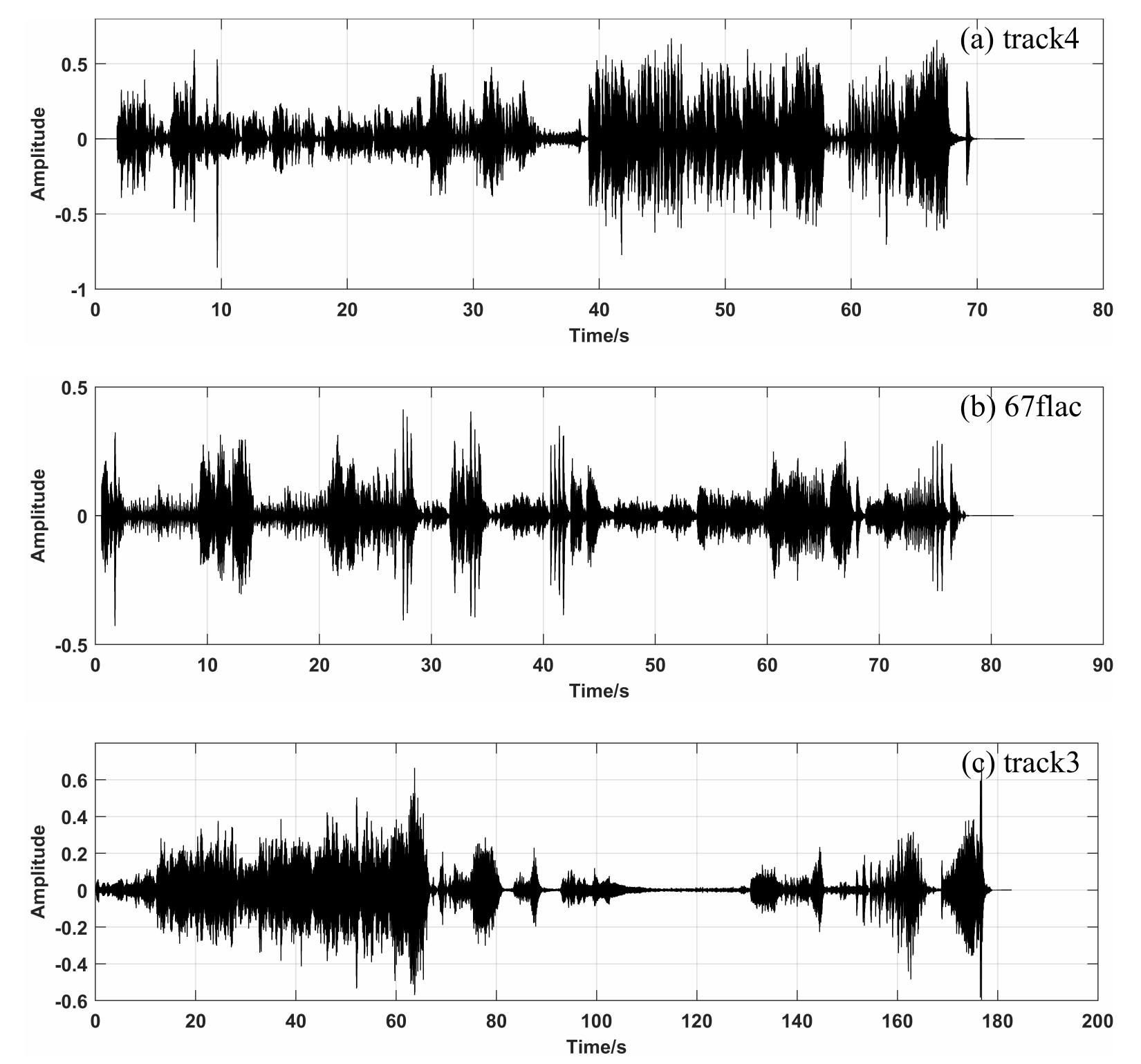}     
	\caption{Representative audio signals in \cite{audio1} and \cite{audio2}.}     
	\label{stiaudio}     
\end{figure}

\section{Experimental Studies}
This section performs numerical simulations to validate the effectiveness and accuracy of the algorithms we proposed. The simulations investigate the scenarios with blind extraction in part A and non-blind extraction in part B separately. The experimental setup is described in detail below.

\textbf{Datasets}: Without loss of generality, we use the images in BOWS-2 \cite{timmurphy.org} database and the audios in \cite{audio1} and \cite{audio2} as the embedding covers. BOWS-2 database consists of $10,000$ grey-level images, with different statistical properties. Fig. \ref{stimage} displays some typical samples, and the labels of the subgraphs indicate their ordinal numbers in the dataset. Audio datasets consist of 9 MP3 and 70 FLAC files, containing different types of audio. Fig. \ref{stiaudio} shows three audio signals samples, and the labels of the subgraphs represents their file name in datasets. 

\textbf{Indicator}: The bit-error-rate (BER), as a common performance index, is employed to measure the extraction performance. 

\textbf{Preprocessing}: By performing $8 \times 8$ block-wise DCT transform, zig-zag scanning and coefficient selection on the original images, we obtain transform domain hosts for embedding. Then invoking additive SS embedding, the watermarked images are generated. Similarly, we can use the same process to generate the watermarked audio signals.

For image cover, the entries in the matrix $\mathbf{V}$ are taken from standard Gaussian distribution; while for audio, the elements in the matrix $\mathbf{V}$ are taken from $\{-1,1\}$. The normalized orthogonality defect of the simulated carriers are shown in Table 1. By varying the size of $L\times K$, the carriers exhibit different $\delta(\overline{\mathbf{V}})$. 
The noise power of image and audio are fixed as $\sigma_{n}^2=3$ and $\sigma_{n}^2=1$ respectively, and the signal-to-noise ratio is controlled by varying the distortion $D$. 

\renewcommand\arraystretch{2.0}
\begin{table}[t!]
	\caption{The normalized orthogonality defect of the simulated carriers.}
	\vspace{-2mm}
	\begin{center}
		\begin{tabular}{c|c|c|c|c|c|c}
			\hline 
			$L\times K$ & $8\times8$ & $12\times10$ & $15\times12$ & $10\times4$ & $12\times6$ & $15\times8$\tabularnewline
			\hline 
			\hline 
			$\delta(\overline{\mathbf{V}})$ & $1.6887$ & $1.3819$ & $1.3475$ & $1.0892$ & $1.1417$ & $1.1695$\tabularnewline
			\hline 
		\end{tabular}
	\end{center}
	\label{tab1}
\end{table}

\begin{figure*}[h!]
	\centering    
	\subfigure[BER versus distortion ($512 \times 512$ No.0, $L=8$, $K=8$, $\delta(\overline{\mathbf{V}})=1.6887$).] { 
		\label{8x8blim}     
		\includegraphics[width=5.62cm]{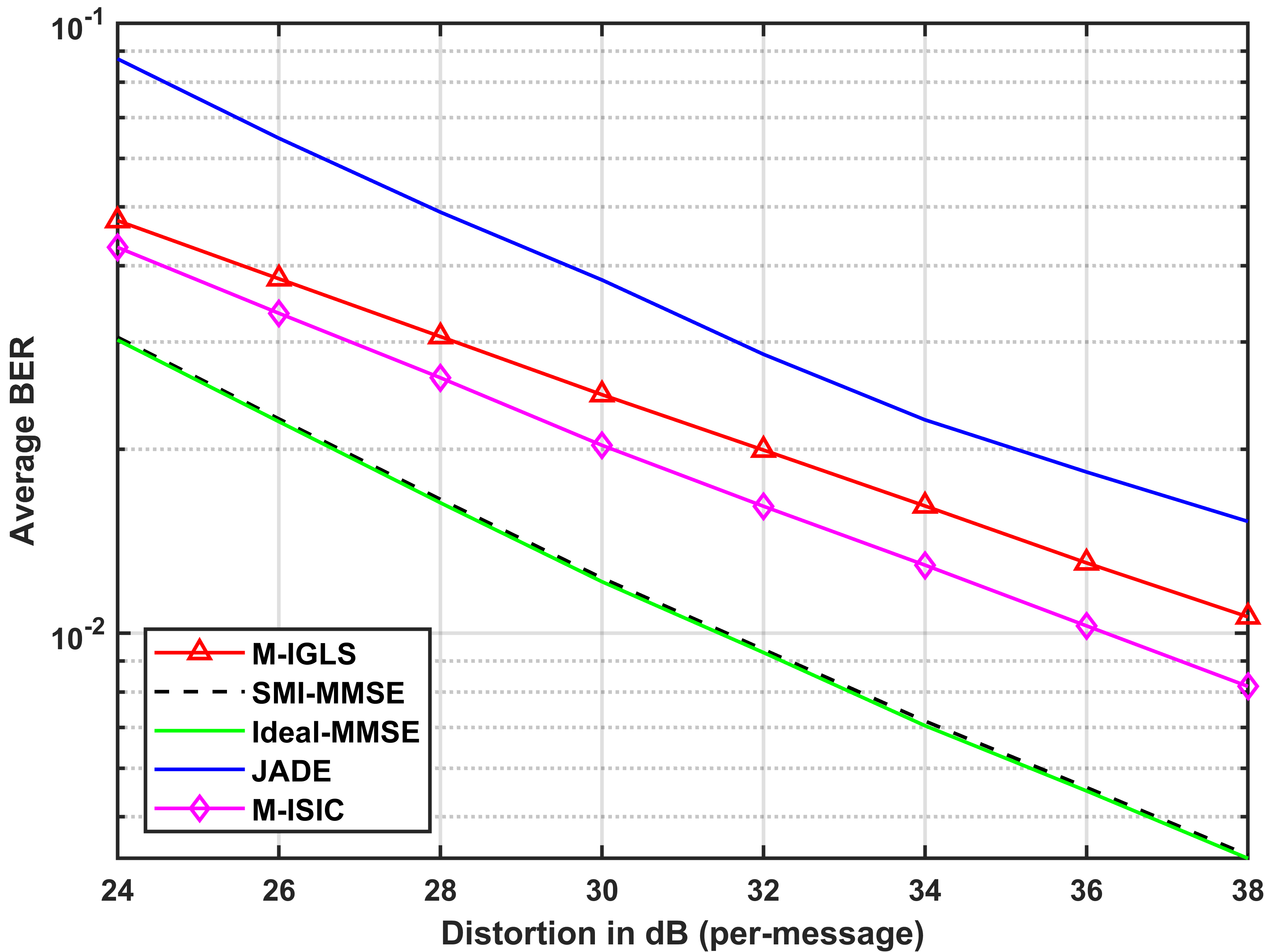}     
	}   
	\subfigure[BER versus distortion ($512 \times 512$ No.709, $L=12$, $K=10$, $\delta(\overline{\mathbf{V}})=1.3819$).] {
		\label{12x10blim}     
		\includegraphics[width=5.62cm]{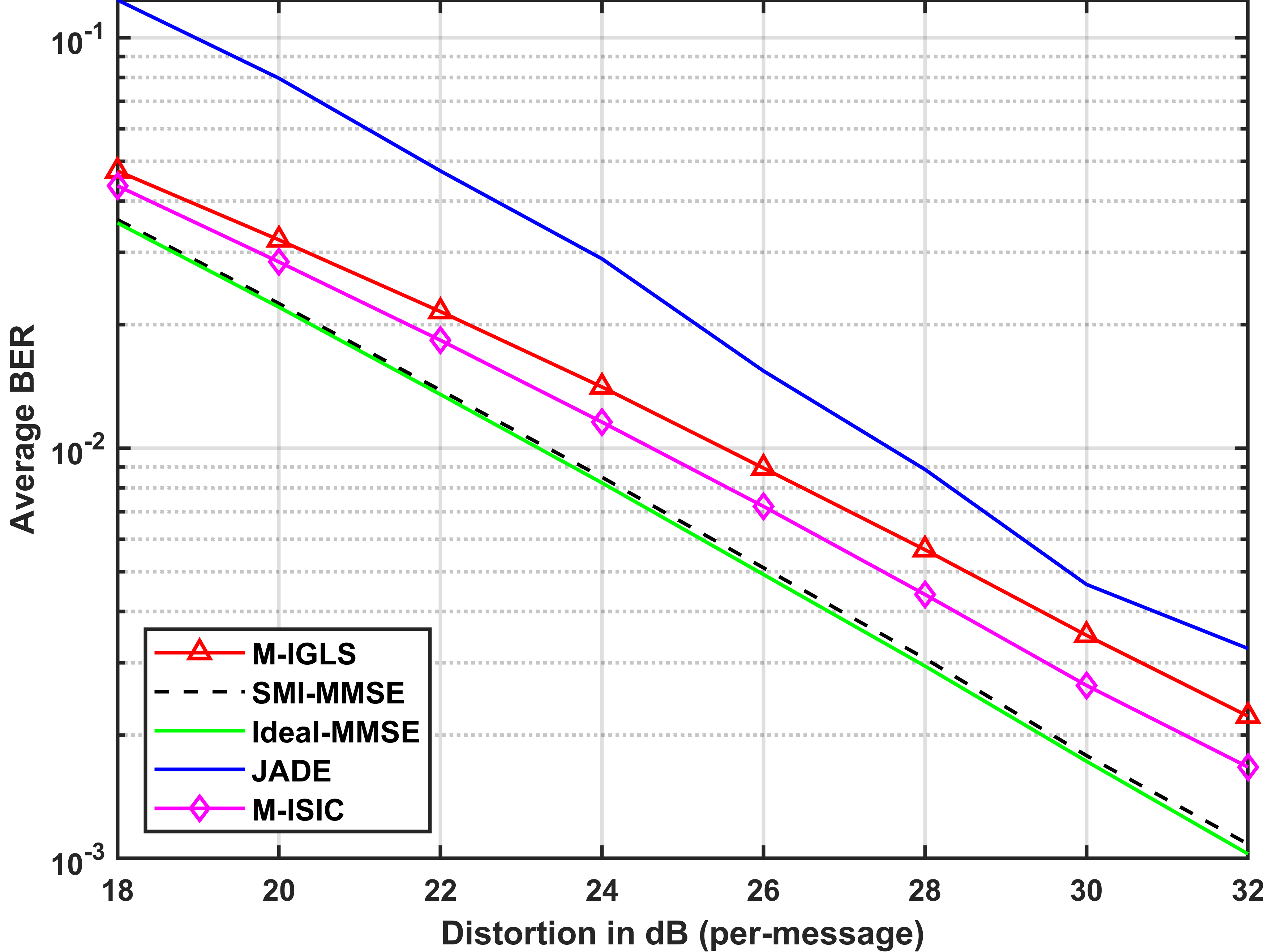}  
	}     
	\subfigure[BER versus distortion ($512 \times 512$ No.5752, $L=15$, $K=12$, $\delta(\overline{\mathbf{V}})=1.3475$).] { 
		\label{15x12blim}     
		\includegraphics[width=5.62cm]{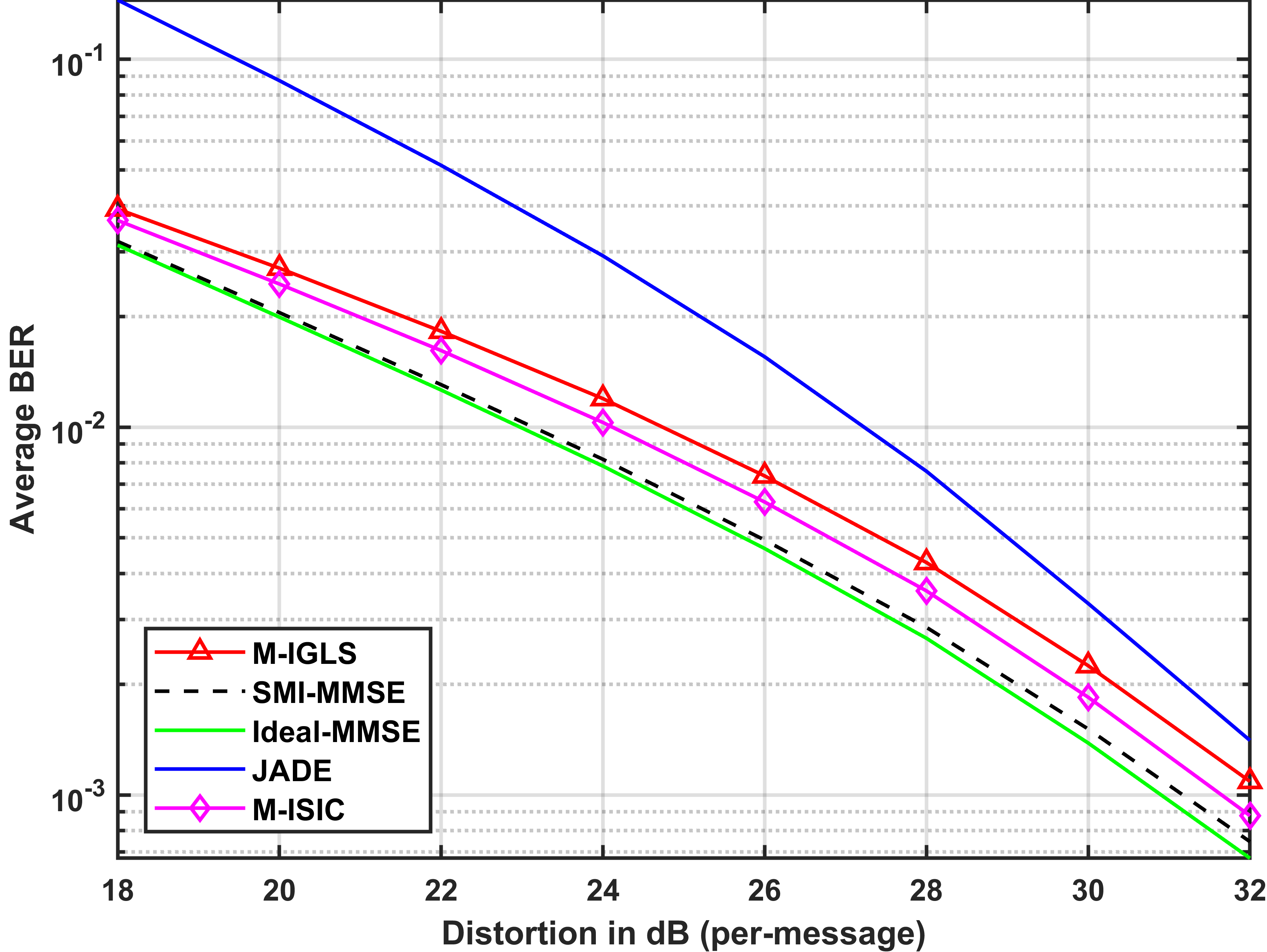}     
	}   
	\caption{BER comparison between M-IGLS and M-ISIC in image blind extraction.}     
	\label{bli_compar}     
\end{figure*}

\begin{figure*}[h!]
	\centering    
	\subfigure[BER versus alphas (track4.mp3, $L=10$, $K=4$, $\delta(\overline{\mathbf{V}})$=1.0892)] { 
		\label{bliaud10x4}     
		\includegraphics[width=5.62cm]{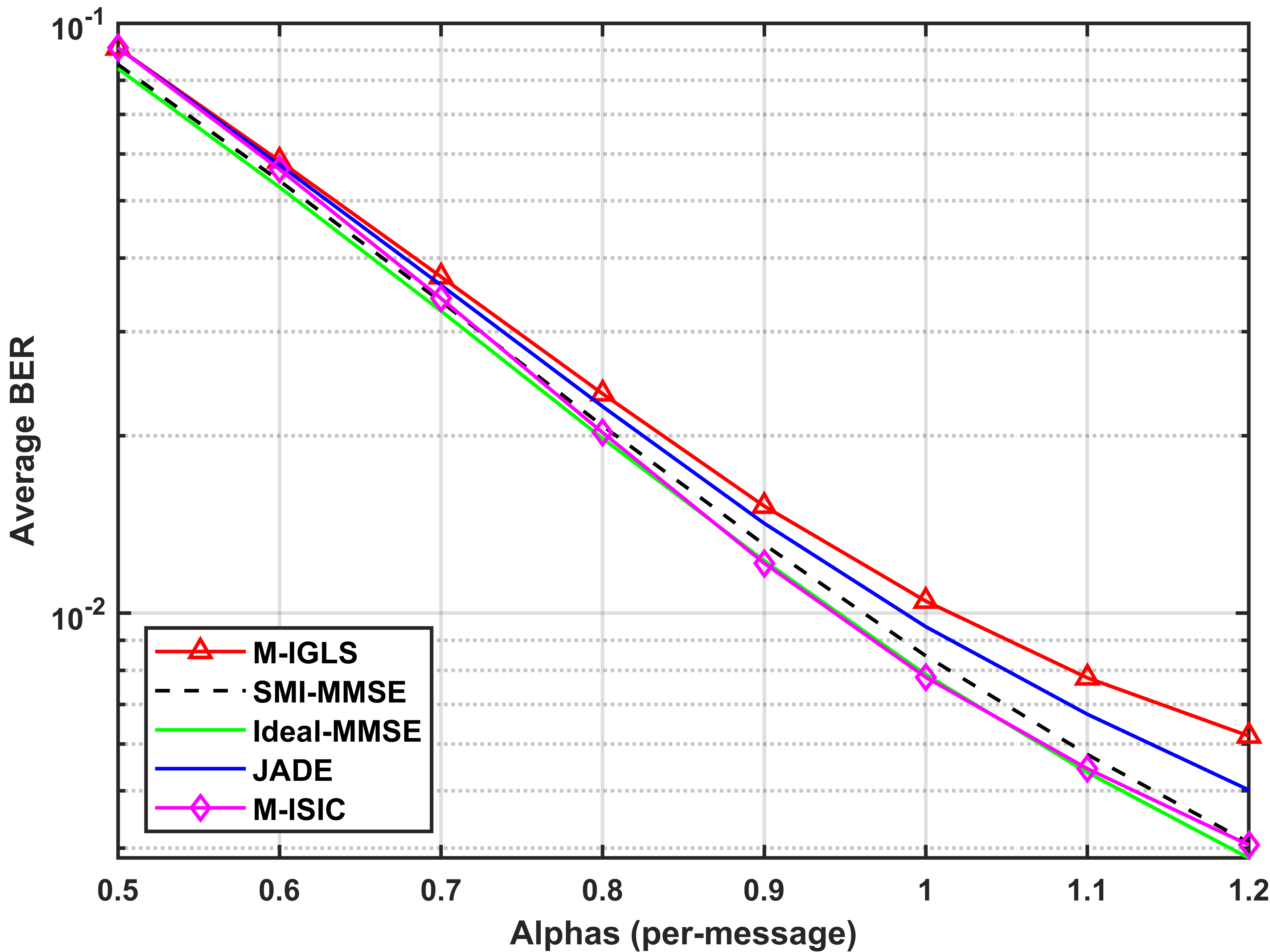}     
	}   
	\subfigure[BER versus alphas (67.flac, $L=12$, $K=6$, $\delta(\overline{\mathbf{V}})$=1.1417)] {
		\label{bliaud12x6}     
		\includegraphics[width=5.62cm]{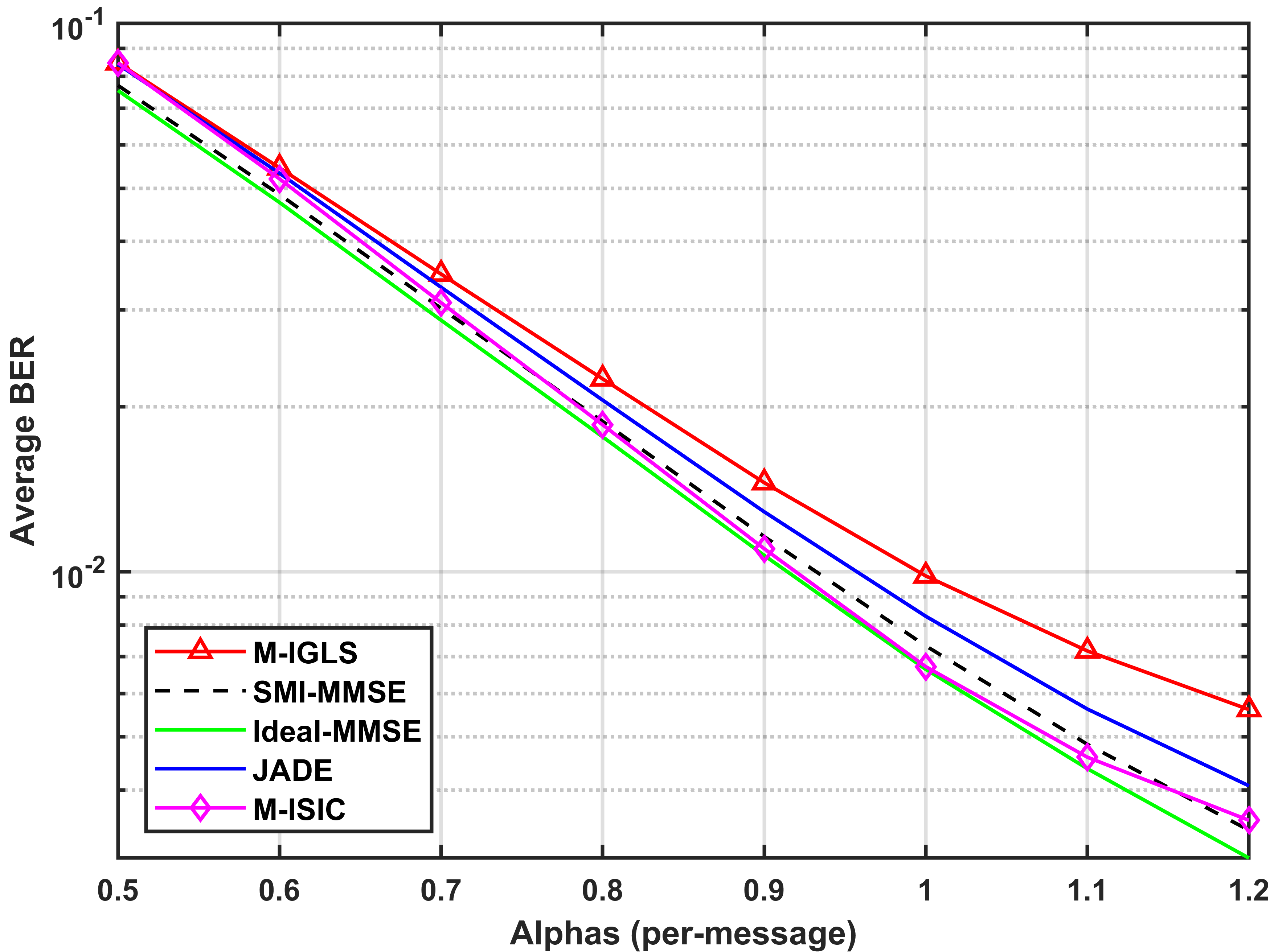}  
	}     
	\subfigure[BER versus alphas (track3.mp3, $L=15$, $K=8$, $\delta(\overline{\mathbf{V}})$=1.1695)] { 
		\label{bliaud15x8}     
		\includegraphics[width=5.80cm]{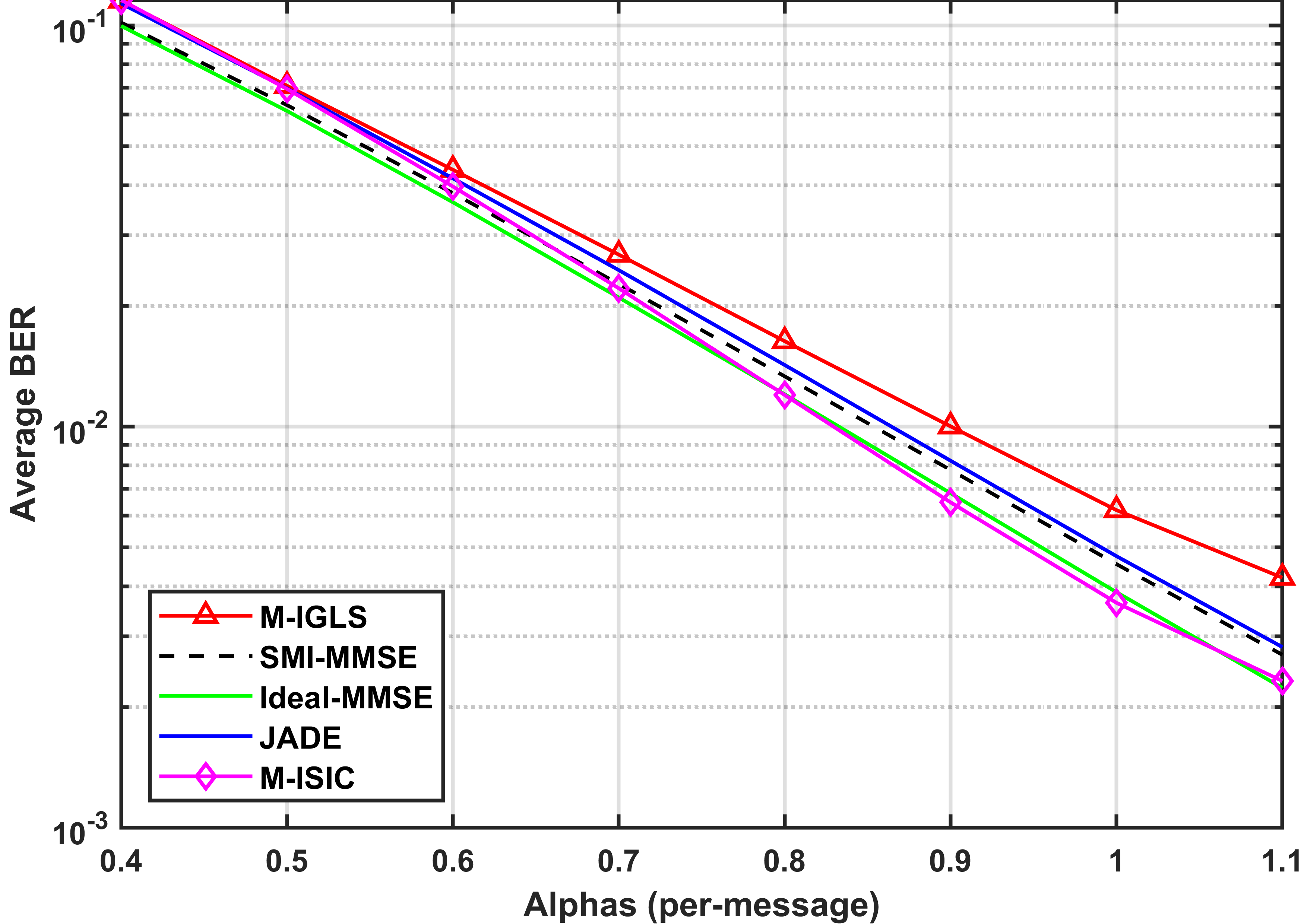}     
	}   
	\caption{BER comparison between M-IGLS and M-ISIC in audio blind extraction.}     
	\label{bli_aud}     
\end{figure*}

\subsection{Blind Extraction}  
Benchmark algorithms in this subsection include: \emph{\romannumeral1)} M-IGLS \cite{li2013extracting}, \emph{\romannumeral2)} SMI-MMSE \cite{ingle2000statisical}, \emph{\romannumeral3)} Ideal-MMSE, \emph{\romannumeral4)} JADE \cite{cardoso1999high}, where Ideal-MMSE represents the ideal version of SMI-MMSE because the autocorrelation matrix   $\mathbf{R_x}$ is exactly known. 	

For image cover, we consider the cases with $\delta(\overline{\mathbf{V}})=1.6887,\,1.3819, \,1.3475$ for the sake of showing the impact of the orthogonality of the lattice bases. In the first example, we consider the case with $L=8$, $K=8$, $\delta(\overline{\mathbf{V}})=1.6887$. The BER versus distortion performance of different algorithms are plotted in Fig. \ref{8x8blim}. With the exact carriers' information, the SMI-MMSE and Ideal-MMSE algorithms serve as the performance upper bounds. The BSS approach, JADE fails to exhibit satisfactory performance. Moreover, M-ISIC outperforms M-IGLS in the whole distortion range of $24\sim 38 \mathrm{dB}$. The second example examines the case with $L=12$, $K=10$, $\delta(\overline{\mathbf{V}})=1.3819$. As depicted in Fig. \ref{12x10blim},
when the carriers become more orthogonal, both M-IGLS and M-ISIC get closer to SMI-MMSE and Ideal-MMSE. The performance gap between M-IGLS and M-ISIC has become smaller. Similar results can be replicated when we further reduce the normalized orthogonality defect. We post one of such figures in Fig. \ref{15x12blim}.

Audio signal is another type of cover for our experiment, with a sampling frequency of 44.1	KHz. We also compare the BER performance for three cases with different quality of carriers, i.e., $\delta(\overline{\mathbf{V}})=1.0892,\,1.1417, \,1.1695$. Fig. \ref{bli_aud} depicts the corresponding experimental results in turn, where the value on horizontal axis controls the distortion degree of audio signal. It is obvious that for the whole alpha range of $0.5\sim 1.2$, M-ISIC has lower BER than M-IGLS in the cases we have listed. From the audio experiments, it can be found that our scheme is superior to M-IGLS even if the value of
 $\delta(\overline{\mathbf{V}})$ is relatively small.

From the above, we observe that M-ISIC performs better than M-IGLS in general. When the carriers $\overline{\mathbf{V}}$ represents a bad lattice basis, M-ISIC apparently outperforms M-IGLS. On the other hand, when the carriers are highly orthogonal, the decision regions of M-IGLS and M-ISIC become similar in shape, then the performance of the two algorithms tends to be the same.

\begin{figure*}[h!]
	\centering    
	\subfigure[BER versus distortion ($512 \times 512$ No.6088, $L=8$, $K=8$, $\delta(\overline{\mathbf{V}})$=1.6887)] { 
		\label{6088}     
		\includegraphics[width=5.62cm]{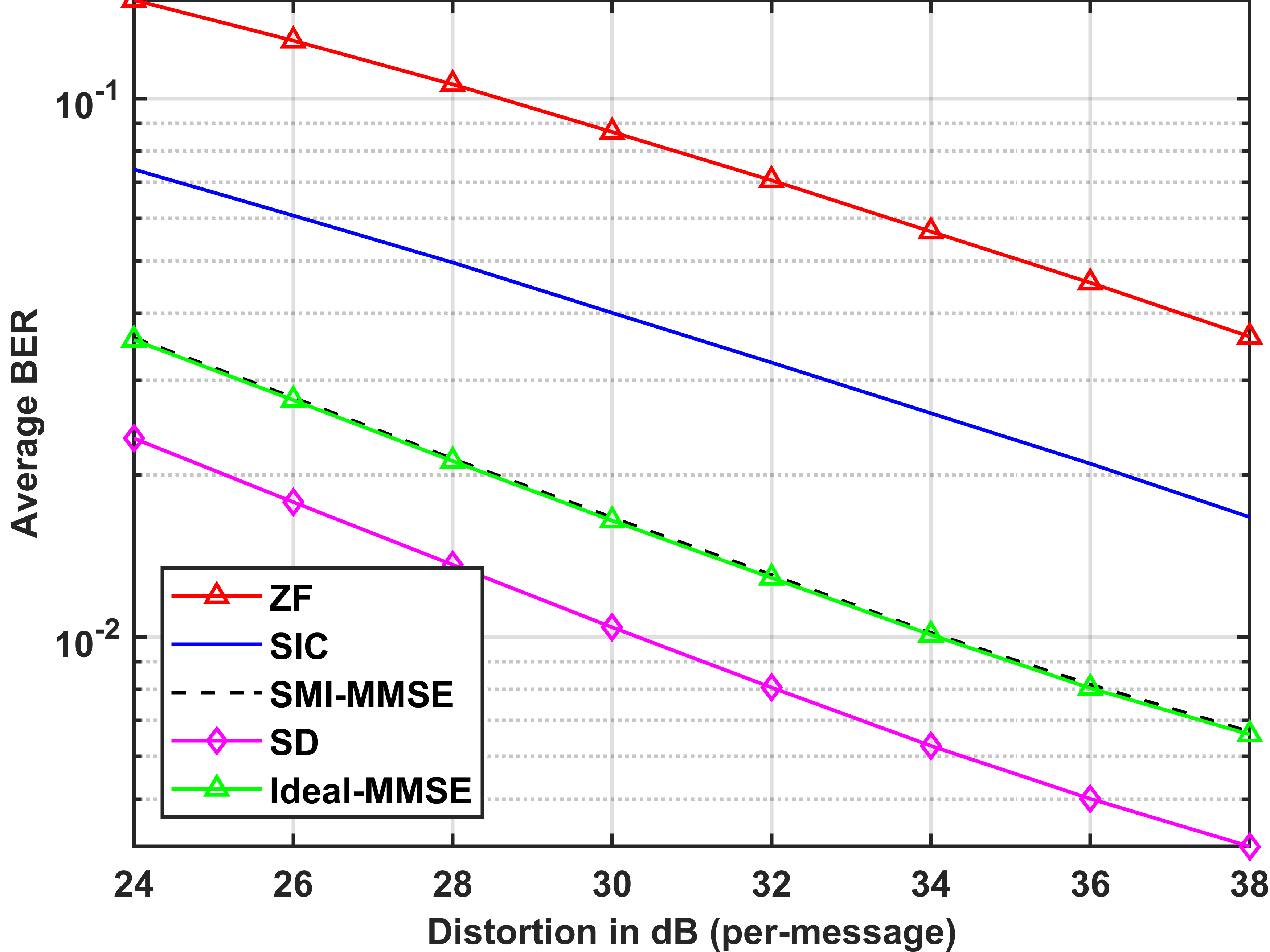}     
	}   
	\subfigure[BER versus distortion ($512 \times 512$ No.6120, $L=12$, $K=10$, $\delta(\overline{\mathbf{V}})$=1.3819)] {
		\label{6120}     
		\includegraphics[width=5.62cm]{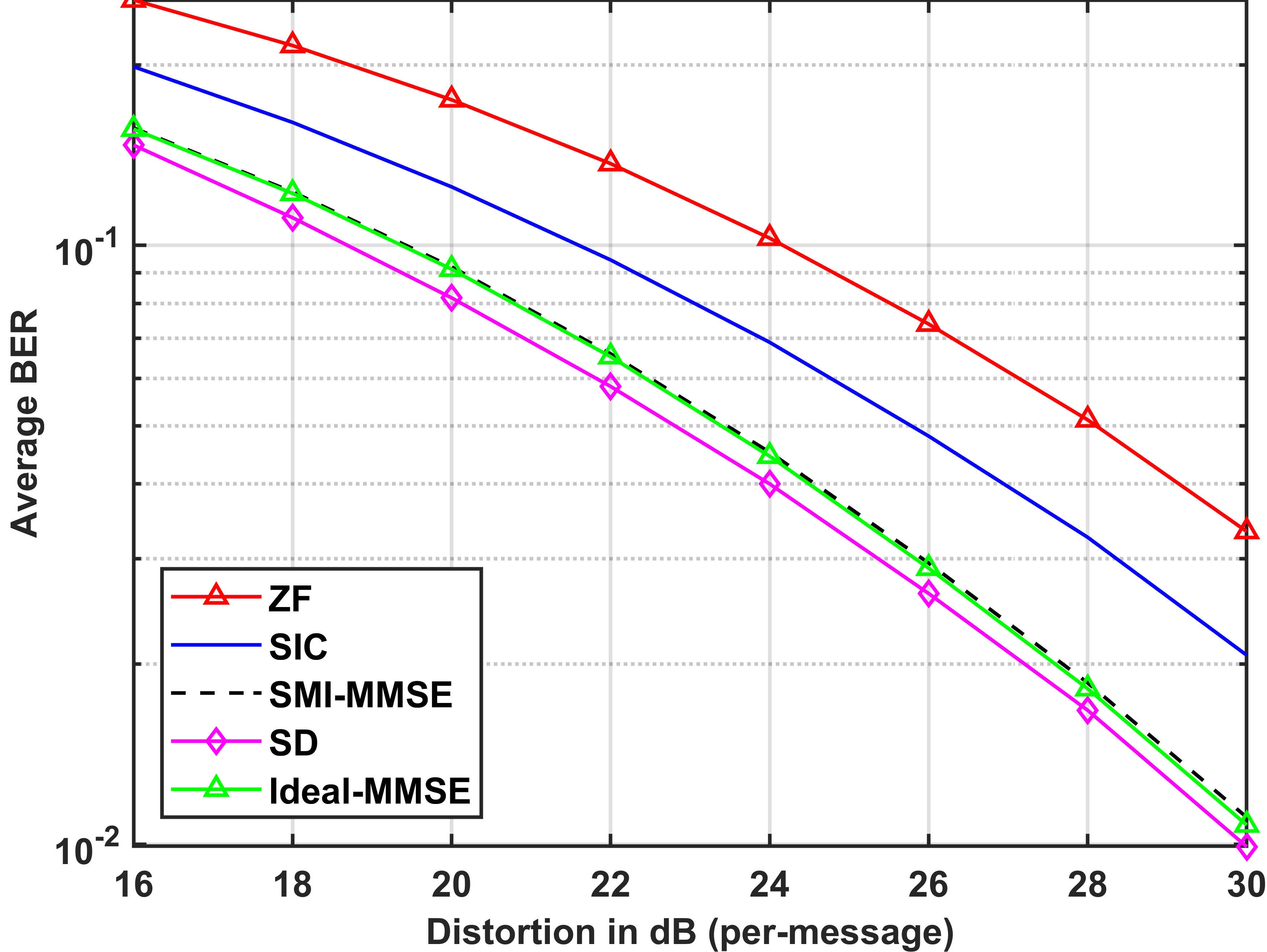}  
	}     
	\subfigure[BER versus distortion ($512 \times 512$ No.6148, $L=15$, $K=12$, $\delta(\overline{\mathbf{V}})$=1.3475)] { 
		\label{6148}     
		\includegraphics[width=5.62cm]{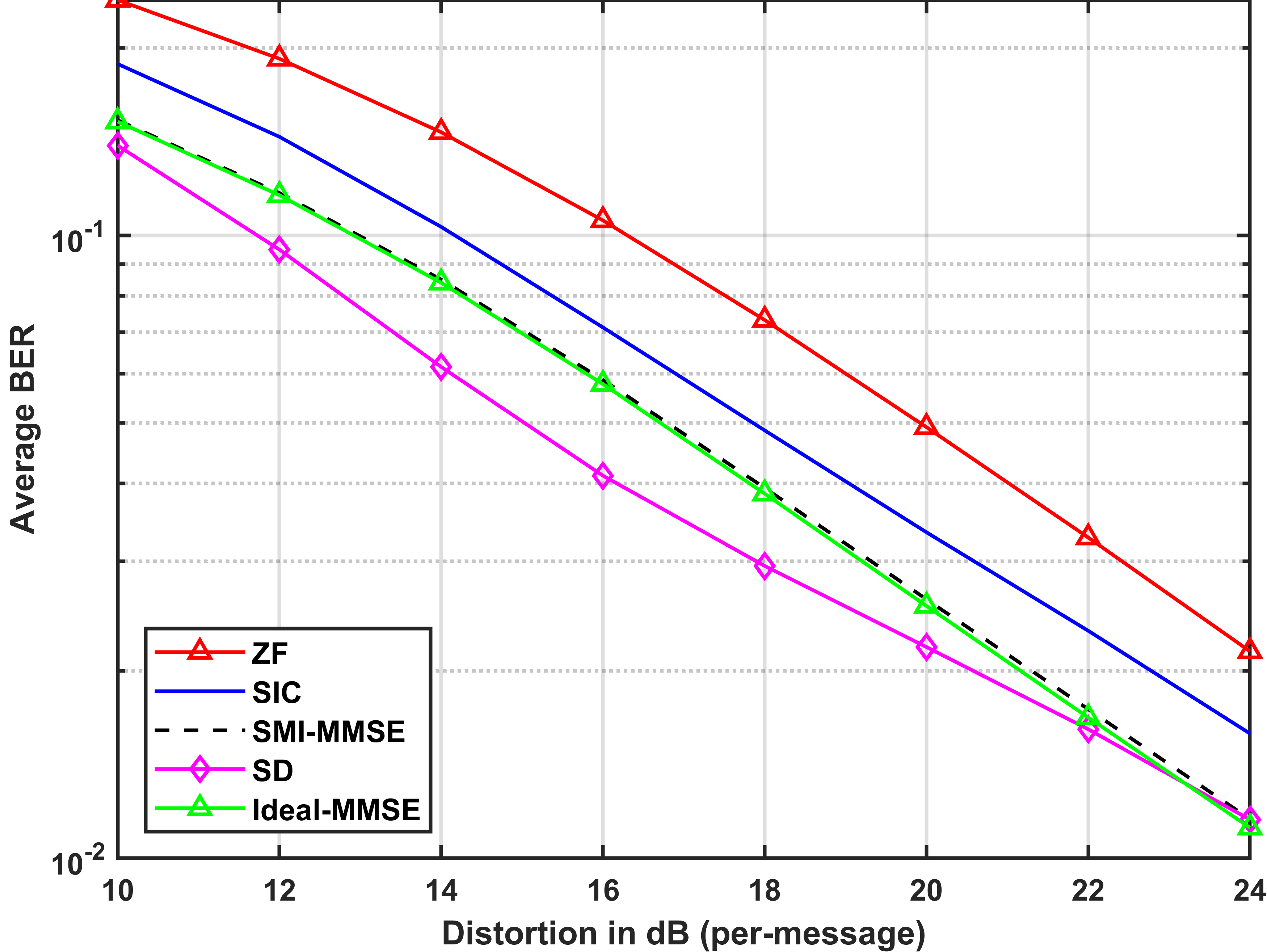}     
	}   
	\caption{BER comparison in image non-blind extraction.}     
	\label{512simu}     
\end{figure*}

 
\subsection{Non-blind Extraction}

When the carriers are known for the receiver, more sophisticated decoding algorithm can be employed to achieve higher accuracy. Next, we examine the BER comparison of each algorithm for non-blind extraction. We adopt the same setting of carriers as in the previous subsection.

From Fig. \ref{6088}, we observe that ZF and SIC struggle to exhibit satisfactory performance due to the bad orthogonality of lattice basis. Sphere decoding significantly outperforms SMI-MMSE and even Ideal-MMSE in the whole distortion range from 24 to 38 dB. If $\delta(\overline{\mathbf{V}})$ decreases, as
shown in Fig. \ref{6120} and Fig. \ref{6148}, the distance between the BER curves of sphere decoding and SMI-MMSE/Ideal-MMSE becomes smaller gradually. However, in the low distortion range, sphere decoding still enjoys the best performance.




The experimental results of audio signals are given in Fig. \ref{audiotest}. In the examples listed below, we find that SIC performs much better on audios rather than on images. Notably, sphere decoding shows a more prominent advantage, especially when alpha increases. Even if the value of $\delta(\overline{\mathbf{V}})$ trends to 1, meaning the basis vectors more orthogonal, the BER of sphere decoding is still much lower than other algorithms.


From the above, we observe that the performance of sphere decoding is generally not worse than that of SMI-MMSE and Ideal-MMSE. When the carriers $\overline{\mathbf{V}}$ represents a bad lattice basis, sphere decoding apparently outperforms SMI-MMSE and Ideal-MMSE. On the other hand, when the carriers are highly orthogonal, the performance of MMSE and SD tends to be the same. 

\begin{figure*}[t]
	\centering    
	\subfigure[BER versus alphas (track4.mp3, $L=10$, $K=4$, $\delta(\overline{\mathbf{V}})$=1.0892)] { 
		\label{12x4audio}     
		\includegraphics[width=5.62cm]{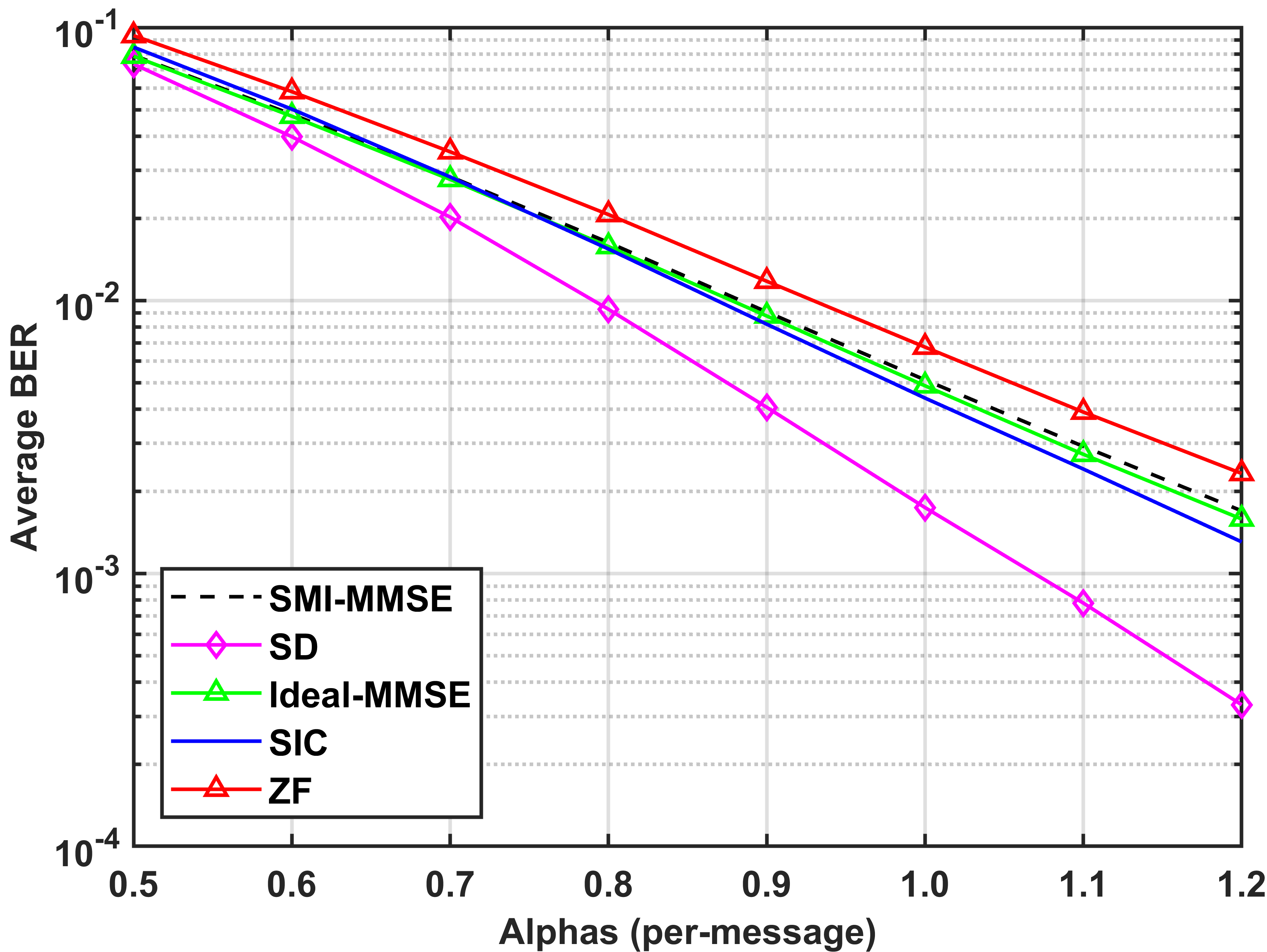}     
	}   
	\subfigure[BER versus alphas (67.flac, $L=12$, $K=6$, $\delta(\overline{\mathbf{V}})$=1.1417)] {
		\label{13x6audio}     
		\includegraphics[width=5.62cm]{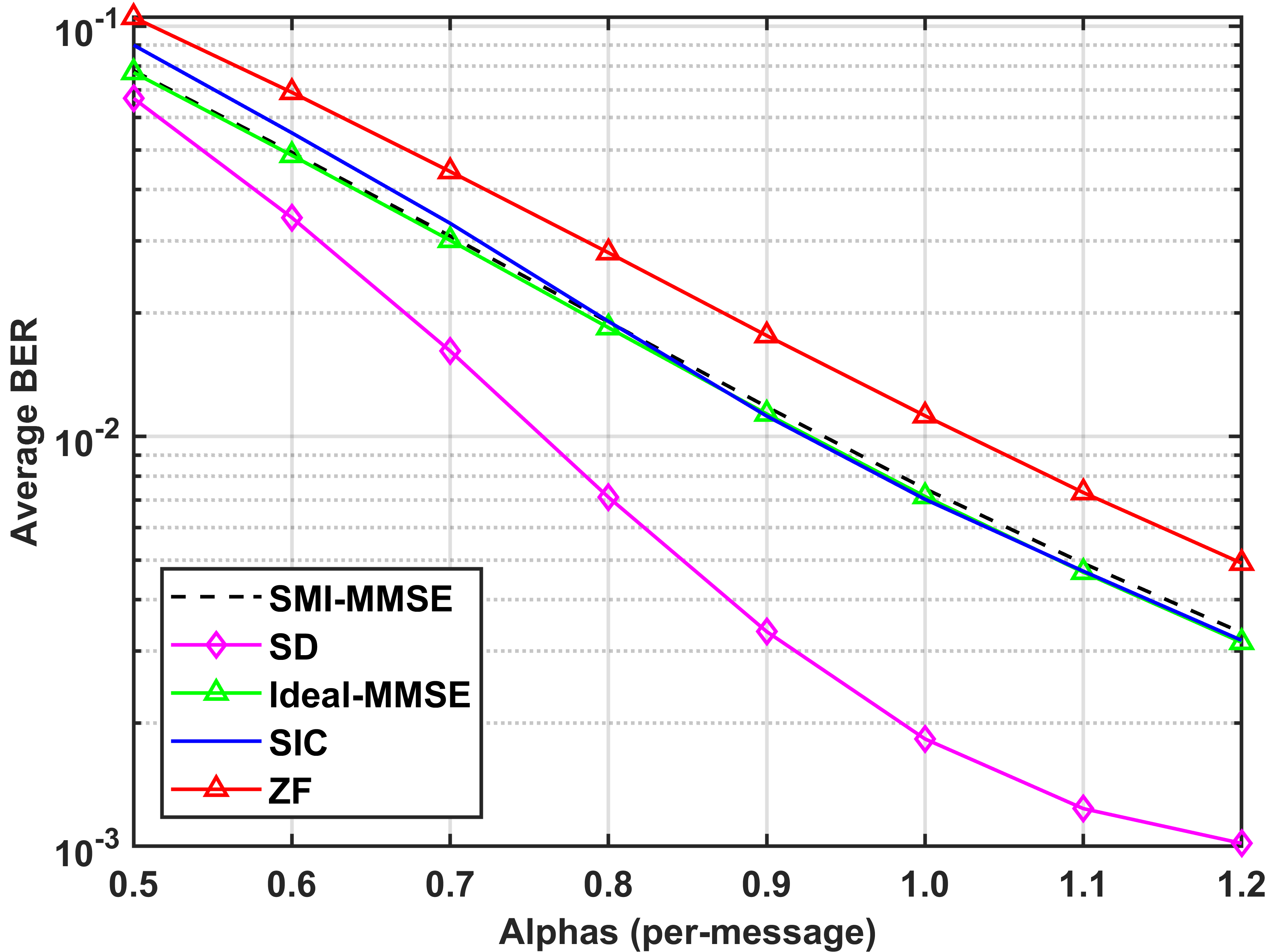}  
	}     
	\subfigure[BER versus alphas (track3.mp3, $L=15$, $K=8$, $\delta(\overline{\mathbf{V}})$=1.1695)] { 
		\label{16x8audio}     
		\includegraphics[width=5.62cm]{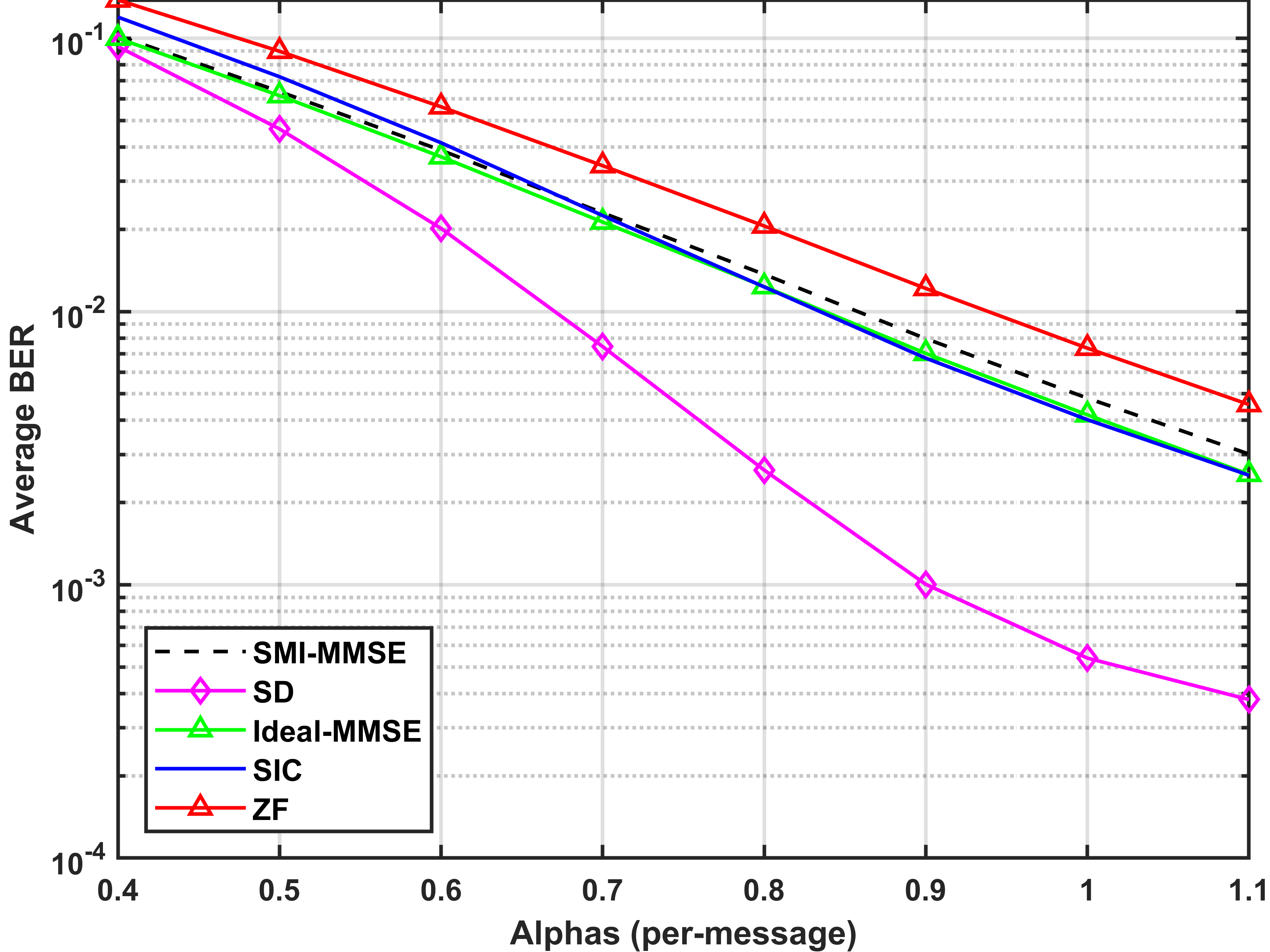}     
	}   
	\caption{BER comparison in audio non-blind extraction.}     
	\label{audiotest}     
\end{figure*}

\section{Conclusions}
This paper studies both blind and non-blind extraction of spread-spectrum hidden data from the perspective of lattices. To achieve better decoding performance, we employ more accurate lattice decoding algorithms in blind and non-blind extraction. The experimental results demonstrate that our schemes are superior to the existing solutions especially when the channel matrix lacks sufficient  orthogonality. 

\appendix
\section{The equivalence of GLS and ZF}
Assuming $\mathbf{V}$ is known, the least-squares estimation \cite{li2013extracting} of $\mathbf{B}$ used in Step 5 of Algorithm 1 is:
\begin{align}
\hat{\mathbf{B}}_{\mathrm{GLS}} &	=  \left(\mathbf{V}^\mathrm{T}\mathbf{R}_{\mathbf{y}}^{-1}\mathbf{V}\right)^\mathrm{-1}
\mathbf{V}^\mathrm{T}\mathbf{R}_{\mathbf{y}}^{-1}\mathbf{Y} \nonumber \\
&=  \left(\left(\mathbf{V}^\mathrm{T}\mathbf{R}_{\mathbf{z}}^{-1}\mathbf{V}\right)^\mathrm{-1} + \mathbf{I} \right)\mathbf{V}^\top \nonumber \\
&\,\,\,\,\,\,\times \left(\mathbf{R}_{\mathbf{z}}^{-1}- \mathbf{R}_{\mathbf{z}}^{-1} \mathbf{V} \left(\mathbf{V}^\mathrm{T}\mathbf{R}_{\mathbf{z}}^{-1}\mathbf{V} + \mathbf{I}\right)^\mathrm{-1} \mathbf{V}^\top \mathbf{R}_{\mathbf{z}}^{-1} \right)   \nonumber\\
& =  \left(\mathbf{V}^\mathrm{T}\mathbf{R}_{\mathbf{z}}^{-1}\mathbf{V}\right)^\mathrm{-1}
\mathbf{V}^\mathrm{T}\mathbf{R}_{\mathbf{z}}^{-1}\mathbf{Y} \nonumber \\
&=\left(\mathbf{V}^\mathrm{T}\mathbf{R}_{\mathbf{z}}^{-\frac{1}{2}}\mathbf{R}_{\mathbf{z}}^{-\frac{1}{2}}\mathbf{V}\right)^\mathrm{-1}\mathbf{V}^\mathrm{T}
\mathbf{R}_{\mathbf{z}}^{-\frac{1}{2}}\mathbf{R}_{\mathbf{z}}^{-\frac{1}{2}}\mathbf{Y} \nonumber \\
&=\left[(\mathbf{R}_{\mathbf{z}}^{-\frac{1}{2}}\mathbf{V})^\mathrm{T}(\mathbf{R}_{\mathbf{z}}^{-\frac{1}{2}}\mathbf{V})\right]^{-1}
(\mathbf{R}_{\mathbf{z}}^{-\frac{1}{2}}\mathbf{V})^\mathrm{T}(\mathbf{R}_{\mathbf{z}}^{-\frac{1}{2}}\mathbf{Y}).\label{GLS33}
\end{align}

In the language of ZF, recall that $\overline{\mathbf{Y}} = \mathbf{R}_{\mathbf{z}}^{-\frac{1}{2}}\mathbf{Y}$, and   $\overline{\mathbf{V}} = \mathbf{R}_{\mathbf{z}}^{-\frac{1}{2}}\mathbf{V}$.  Thus Eq. (\ref{GLS33}) equals to $(\overline{\mathbf{V}}^\mathrm{T}\overline{\mathbf{V}})^{-1}\overline{\mathbf{V}}^\mathrm{T}\overline{\mathbf{Y}}^\mathrm{T}$, which justifies $\hat{\mathbf{B}}_{\mathrm{GLS}} = \hat{\mathbf{B}}_{\mathrm{ZF}}$.

\bibliographystyle{IEEEtranMine}
\bibliography{lib}

\end{document}